\documentclass{article}
\usepackage{amsmath,mathtools}

\newcommand{\xddots}{%
  \raise 4pt \hbox {.}
  \mkern 6mu
  \raise 1pt \hbox {.}
  \mkern 6mu
  \raise -2pt \hbox {.}
}

\usepackage{alltt}
\usepackage[vcentermath]{youngtab}
\usepackage[all]{xy}
\usepackage{latexsym,amssymb,amsmath,amsfonts, amsthm, xcolor}

\newcommand{\CC}{\mathbb{C}}

\def\vbar{\mathchoice{\vrule height6.3ptdepth-.5ptwidth.8pt\kern-.8pt}
  {\vrule height6.3ptdepth-.5ptwidth.8pt\kern-.8pt}
  {\vrule height4.1ptdepth-.35ptwidth.6pt\kern-.6pt}
  {\vrule height3.1ptdepth-.25ptwidth.5pt\kern-.5pt}}
\def\fudge{\mathchoice{}{}{\mkern.5mu}{\mkern.8mu}}
\def\bbc#1#2{{\rm \mkern#2mu\vbar\mkern-#2mu#1}}
\def\bbb#1{{\rm I\mkern-3.5mu #1}}
\def\bba#1#2{{\rm #1\mkern-#2mu\fudge #1}}
\def\bb#1{{\count4=`#1 \advance\count4by-64 \ifcase\count4\or\bba A{11.5}\or
  \bbb B\or\bbc C{5}\or\bbb D\or\bbb E\or\bbb F \or\bbc G{5}\or\bbb H\or
  \bbb I\or\bbc J{3}\or\bbb K\or\bbb L \or\bbb M\or\bbb N\or\bbc O{5} \or
  \bbb P\or\bbc Q{5}\or\bbb R\or\bbc S{4.2}\or\bba T{10.5}\or\bbc U{5}\or
  \bba V{12}\or\bba W{16.5}\or\bba X{11}\or\bba Y{11.7}\or\bba Z{7.5}\fi}}

\newtheorem{theorem}{Theorem}
\newtheorem{itlemma}{Lemma}[section]
\newtheorem{itproposition}[itlemma]{Proposition}
\newtheorem{itcorollary}[itlemma]{Corollary}
\newtheorem{itremark}[itlemma]{Remark}
\newtheorem{itremarks}[itlemma]{Remarks}
\newtheorem{itdefinition}[itlemma]{Definition}
\newtheorem{itexample}[itlemma]{Example}

\newenvironment{lemma}{\begin{itlemma}\rm}{\end{itlemma}} 
\newenvironment{remark}{\begin{itremark}\rm}{\end{itremark}} 
\newenvironment{remarks}{\begin{itremarks} \rm}{\end{itremarks}}
\newenvironment{corollary}{\begin{itcorollary}\rm}{\end{itcorollary}}
\newenvironment{proposition}{\begin{itproposition}\rm}{\end{itproposition}}
\newenvironment{definition}{\begin{itdefinition}\rm}{\end{itdefinition}}
\newenvironment{example}{\begin{itexample}\rm}{\end{itexample}}
\newenvironment{fact}{\noindent {{\bf Fact}}:\ \ }{\hfill \medskip}
\newenvironment{claim}{\noindent {\em Claim}. \ \ }{\hfill \medskip}
\newcommand{\be}[1]{\begin{equation}\label{#1}}
\newcommand{\ee}{\end{equation}}
\newcommand{\bl}[1]{\begin{lemma}\label{#1}}
\newcommand{\br}[1]{\begin{remark}\label{#1}}
\newcommand{\brs}[1]{\begin{remarks}\label{#1}}
\newcommand{\bt}[1]{\begin{theorem}\label{#1}}
\newcommand{\bd}[1]{\begin{definition}\label{#1}}
\newcommand{\bp}[1]{\begin{proposition}\label{#1}}
\newcommand{\bc}[1]{\begin{corollary}\label{#1}}
\newcommand{\bfact}[1]{\begin{fact}\label{#1}}
\newcommand{\bex}[1]{\begin{example}\label{#1}}
\newcommand{\ec}{\end{corollary}}
\newcommand{\efact}{\end{fact}}
\newcommand{\eex}{\end{example}}
\newcommand{\el}{\end{lemma}}
\newcommand{\er}{\end{remark}}
\newcommand{\ers}{\end{remarks}}
\newcommand{\et}{\end{theorem}}
\newcommand{\ed}{\end{definition}}
\newcommand{\ep}{\end{proposition}}
\newcommand{\epr}{\end{proof}}
\newcommand{\bpr}{\begin{proof}}
\newcommand{\bcl}{\begin{claim}}
\newcommand{\ecl}{\end{claim}}

\newcommand{\bi}{\begin{itemize}}
\newcommand{\ei}{\end{itemize}}
\newcommand{\ben}{\begin{enumerate}}
\newcommand{\een}{\end{enumerate}}

\usepackage{graphicx}

\title{Symmetric States and Dynamics of Three Quantum Bits}
\author{Francesca Albertini\thanks{Dipartimento di Tecnica e Gestione del Sistemi Industriali, Universita' di Padova, Vicenza, Italy, albertin@math.unipd.it}  \and Domenico D'Alessandro\thanks{Corresponding Author, Department of Mathematics, Iowa State University, Ames IA 50011, U.S.A., daless@iastate.edu}}
\date{\today}

\begin{document}

\maketitle

\begin{abstract}

The unitary group acting on the Hilbert space ${\cal H}:=(\CC^2)^{\otimes 3}$ of three quantum bits admits a Lie subgroup, $U^{S_3}(8)$,  of elements which permute with the symmetric group of permutations. 
Under the action of such Lie subgroup, the Hilbert space ${\cal H}$ splits into three invariant subspaces of dimensions $4$, $2$ and $2$ respectively, each corresponding to an irreducible representation of $su(2)$. The subspace of dimension $4$ is uniquely determined and corresponds to states that are themselves invariant under the action of the symmetric group. This is the so called  {\it symmetric sector.}

We provide an analysis of  pure states in the symmetric sector  of three quantum bits for what concerns their entanglement properties, separability criteria and dynamics under the Lie subgroup $U^{S_3}(8)$. We parametrize all the possible invariant two-dimensional 
subspaces and extend the previous analysis to these subspaces as well. We propose a physical set up for the states and dynamics we study which consists of  a symmetric network of three spin $\frac{1}{2}$ particles under a common driving electro-magnetic field. {\color{black} For such set up, we solve a control theoretic problem which consists of driving a separable state to a state with maximal distributed entanglement.}

\end{abstract}

\vspace{0.25cm}

\vspace{0.25cm}

\vspace{0.25cm}

\vspace{0.25cm}


\section{Introduction} 

The study of quantum states is a current line of research in quantum physics (see, e.g., \cite{ZK}), in particular for what concerns their entanglement properties. Entanglement is considered a resource in quantum information processing and classifying states according to the amount and type of entanglement is a problem of both fundamental and practical importance. A related problem is to study how quantum dynamical  evolution changes the entanglement of states (see,e g., \cite{ZS} for the two qubits case) and in the quantum control setting \cite{Mikobook} how to induce such dynamics in a specific physical setup. The two qubits case is fairly well understood, while the case of three qubits requires further exploration. In particular,  three qubits are the simplest type of  systems which display two types of entanglement: a {\it pairwise entanglement}  quantifying the entanglement between pairs of qubits and a {\it distributed entanglement}  \cite{Coff}.

In this paper we are concerned with three qubits systems whose dynamics are subject to a permutation symmetry among the three qubits. The possible unitary evolutions on the Hilbert space ${\cal H}:=(\CC^2)^{\otimes 3}$ of the system consists of unitaries  which commute with the permutation group $S_3$. Such a Lie subgroup of $U(8)$, which we denote by $U^{S_3}(8)$ has dimension $20$ \cite{noisym1}. Its Lie algebra,  $u^{S_3}(8)$,  is spanned by the  matrices 
\be{elem}
i \Pi (\sigma_1 \otimes \sigma_2 \otimes \sigma_3).  
\ee
where $\Pi$ denotes the  {\it symmetrization operator} $\Pi:=\frac{1}{3!} \sum_{P \in S_3} P$ and $\sigma_{1,2,3}$ are chosen to be the $2 \times 2$ identity or one of the Pauli matrices 
\be{PauliMat}
\sigma_x:=\begin{pmatrix} 0 & 1 \cr 1 & 0  \end{pmatrix}, \qquad \sigma_y:=
 \begin{pmatrix} 0 & i \cr -i & 0  \end{pmatrix}, \qquad
 \sigma_z:= \begin{pmatrix}  1 & 0 \cr 0 & -1\end{pmatrix}.
\ee
The number of the matrices (\ref{elem}) is equal to the way of choosing the number of occurrences of the identity and $\sigma_{x,y,z}$ out of three positions, which is equal to $20$.\footnote{The general argument for $n$ qubit is presented in \cite{noisym1}.}

The three qubit subspace ${\cal H}$ under the action of the Lie algebra $u^{S_3}(8)$, or of the Lie group $U^{S_3}(8)$,  splits into three invariant subspaces two of which have dimension $2$ and one of which has  dimension $4$. These correspond to irreducible representations of $su(2)$ \cite{noisym1} \cite{Sym3}. The subspace of dimension $4$ is uniquely determined. It is  the so-called {\it symmetric sector} \cite{Ribeiro}, that is, the space 
of the states that do not change under permutation of two qubits. For the three qubits case, we shall use the orthogonal  basis (not normalized)
\be{earray}
\phi_0:=|000\rangle 
\ee
$$
\phi_1:=  |100\rangle + |010\rangle +|001\rangle 
$$
$$
\phi_2:=  |110\rangle + |110\rangle +|101\rangle \nonumber
$$
$$
\phi_3:=|111\rangle. 
$$
This notation follows  the number of $1$'s that appear in each state, in the sense that $\phi_j$ is the sum of the states in the computational basis  which have $j$ $1$'s. States in the symmetric tensor where studied in \cite{Ribeiro} and a complete list of  invariants under local unitary and symmetric transformations  was given there.

We will write a general state in the symmetric sector as 
\be{psibas}
\psi:=c_0 \phi_0+ c_1 \phi_1 + c_2 \phi_2 + c_3 \phi_3, 
\ee 
for complex coefficients $c_0, c_1, c_2, c_3$, with $ |c_0|^2+3 |c_1|^2
+ 3 |c_2|^2+ |c_3|^2=1$. This can also be seen as a  four level system which can be used to implement two qubits.  

The subspaces of dimension 2, which represent isomorphic representations of $su(2)$ are {\it  not} uniquely. . 

The results presented in this paper and a plan for the following sections is as follows. In section \ref{Para2} we perform the decomposition of the Hilbert space ${\cal H}:=\CC^{\otimes 3}$ into invariant subspaces for $u^{S_3}(8)$.  We see that while the four dimensional invariant subspace is uniquely determined (it is the symmetric sector)  the  pair of $2$-dimensional invariant subspaces is  not. We obtain a parametrization of  all the possible decompositions.  In section \ref{entgen} we recall the general measures of entanglement for three qubits introduced in \cite{Coff}, in particular the {\it pairwise entanglement} (concurrence), quantifying the entanglement between two qubits when the third one 
is traced out and the {\it distributed entanglement}. These measures are invariant under 
local unitary transformations. In section \ref{ESS},  we calculate the entanglement measures for states in the symmetric sector and give conditions of separability. The expressions we find complement the ones found in \cite{Ribeiro} which are based on the Majorana polynomial representations of states \cite{Makele}. We also briefly recall such a representation which has an elegant geometric representation and provide a complete set of local invariant for these states. In section \ref{E2S} we analyze the entanglement of the states in the  two dimensional invariant subspaces. 
We then turn our attention to the {\it dynamics} on the invariant subspaces.  The dynamics we are interested in are the ones in the Lie subgroup of the unitary group $U(8)$ which permutes with the symmetric group, i.e., $U^{S_3}(8)$  and that, 
 as we have recalled, leaves the above subspaces invariant.   In section \ref{SSD1} we study the dynamics on the invariant subspaces. In particular for the symmetric sector we prove that the group of local (symmetric) unitary transformations is a maximal  Lie subgroup of $U^{S_3}(8)$ which leaves the adopted measures of entanglement unchanged. We then study in general how the elements of the  group $U^{S_3}(8)$ change the entanglement in this subspace. {\color{black} We give a factorization of possible unitary evolutions on the symmetric sector in evolutions  that modify the entanglement and  evolutions that do not.}  In section \ref{App} we give  a physical application of 
 the analysis described in the previous sections. We consider a symmetric network of three spin $\frac{1}{2}$ particles coupled via identical Ising interaction and driven by a common electro-magnetic field. The dynamics of this model satisfies the symmetry assumptions   considered in this paper. We propose an algorithm to drive such a system from a separable state (with zero entanglement) to  a state with maximum distributed entanglement. A summary of the results is given in section \ref{Summa}.

\section{Decomposition into Invariant Subspaces}\label{Para2}

Consider the Lie algebra $su^{S_3}(8)$, that is,  the subalgebra of $su(8)$ of matrices which commute with the permutation matrices in $S_3$. For $su(8)$,  we consider the standard representation on $\CC^8\simeq {\cal H}$ and therefore the matrices in $su^{S_3}(8)$ are also $8 \times 8$. In appropriate coordinates such matrices take a block diagonal form with blocks of dimension $4\times 4$, $2\times 2$ and $2 \times 2$,  which correspond to invariant subspaces of  ${\cal H}$ of dimensions  $4$, $2$ and $2$, respectively \cite{noisym1}. Such subspaces correspond to irreducible representations of $su(2)$ of dimensions $4$, $2$ and $2$ respectively (the ones of dimension $2$ being isomorphic representations).  To obtain a basis for such subspaces, in terms of the  computational  basis $\{|  jkl\rangle \, \, ,  j,k,l=0,1 \}$  one may  apply standard methods of the quantum theory of angular 
momentum (see, e.g., \cite{Hamarmesh}, \cite{Sakurai}) which overlap with representation theory and the theory of Young tableau and representations of the symmetric group (see, e.g., \cite{Boh}). For instance, the Clebsch-Gordan coefficients described in \cite{Hamarmesh} (in the table on pg. 375), give one possible change of coordinates   to obtain the bases of such invariant subspaces. Another method is given by the use of Young symmetrizers which was reviewed in \cite{Sym3} in the form that uses {\it Hermitian} Young symmetrizers as described in \cite{Zeili}. 
It must be stressed however that such a decomposition in invariant subspaces is not unique and it is of interest to find {\it all} the possible decompositions. In order to achieve this goal we will  use in this section a technique that was described in \cite{Mikobook} (see Chapter 4 section 4.3.4).
One considers the {\it commutant} ${\cal C}$ of $su^{S_3}(8)$ in $u(8)$ which is a reductive Lie algebra and therefore it admits Cartan subalgebras, i.e., a maximal Abelian subalgebra. 
The main observation is that, if $W \oplus V_1 \oplus V_2$ (with $\dim(W)=4$ and $\dim(V_1)=\dim(V_2)=2$)  is a decomposition in invariant subspaces for $su^{S_3}(8)$, then ${\cal C}$ admits a Cartan subalgebra which,  in the appropriate coordinates,  has a basis given by $A_1:=\texttt{diag}(i{\bf 1}_4, {\bf 0}_2, {\bf 0}_2)$, $A_2:=
\texttt{diag}({\bf 0}_4,i{\bf 1}_2, {\bf 0}_2)$,    $A_3:=
\texttt{diag}({\bf 0}_4, {\bf 0}_2, i{\bf 1}_2)$, where ${\bf 1}_r$ (${\bf 0}_r$) is the $r \times r$ identity (zero) matrix, and $\texttt{diag}$ here refers to block diagonal matrices.  Possible decompositions are therefore in correspondence with Cartan subalgebras of the commutant ${\cal C}$. Thus,  a method to obtain all the possible decompositions is the following algorithm. 
\begin{enumerate}
\item Compute the commutant ${\cal C}$ of $su^{S_3}(8)$ in $u(8)$. 
\item Find all possible Cartan subalgebras of ${\cal C}$ which (in this case) all have dimension $3$. 
The following steps refers to the a Cartan subalgebra ${\cal A}$. In order to deal with Hermitian matrices rather than skew-Hermitian ones, we consider $i{\cal A}$.  
\item Take  a basis of $i{\cal A}$ and (orthogonally) diagonalize its  elements  simultaneously (this is possible since these are mutually commuting Hermitian matrices). 
\item Place the elements on the diagonal in three row vectors so as to form a $3 \times 8$ matrix, which we shall denote by $M$.

\item Perform a Gaussian row reduction algorithm to place  the matrix in a Reduced Row Echelon Form (see, e.g., \cite{LinAlg}). This corresponds to taking linear combinations of the matrices in the basis if $i{\cal A}$ so as to obtain a new basis of elements which only have eigenvalues $1$ and $0$. Call these elements (in the original coordinates) 
$(\tilde A_1, \tilde A_2, \tilde A_3)$ with $\tilde A_1$ having eigenvalue $1$ with multiplicity 4, $\tilde A_2$ and $\tilde A_3$ having eigenvalue $1$ with multiplicity 2 

\item $W$ is the eigenspace of $\tilde A_1$ corresponding to eigenvalue $1$. $V_1$ and $V_2$ are the eigenspaces of $\tilde A_2$ and $\tilde A_3$, respectively, corresponding to eigenvalue $1$. 
Notice that once we know $W$ and $V_1$, the subspace $V_2$ is simply  the orthogonal complement. 

\end{enumerate}

Let us carry out the above program for our example.  The commutant ${\cal C}$ is found by solving the linear system of equations $[{\cal C}, B_j]=0$, where $\{B_j\}$ is a basis of $su^{S_3}(8)$. In fact, since the matrices $iH_x$, $iH_y$ and $iH_{zz}$, with (cf. (\ref{PauliMat}))
\be{Hxy}
H_{x,y,z}:=\sigma_{x,y,z} \otimes {\bf 1} \otimes {\bf 1}+ {\bf 1} \otimes \sigma_{x,y,z} \otimes {\bf 1} + {\bf 1} \otimes {\bf 1} \otimes \sigma_{x,y,z}, 
\ee 
\be{Hzz}
H_{zz}:=\sigma_z \otimes \sigma_z \otimes {\bf 1} + {\bf 1} \otimes \sigma_z \otimes \sigma_z + \sigma_z \otimes {\bf 1} \otimes \sigma_z,  
\ee
 generate all of $su^{S_3}(8)$  
\cite{noisym1}, it is enough to solve 
$$
[{\cal C}, H_x]=0, \qquad [{\cal C}, H_y]=0, \qquad  [{\cal C}, H_{zz}]=0. 
$$
This computation,  which was done in \cite{Mikobook}, leads to the  basis 
$\{E_1,E_2,E_3,E_4, E_5\}$ for ${\cal C}$, with  
$$
iE_1:={\bf 1} \otimes {\bf 1} \otimes {\bf 1},$$
$$ iE_2:=\sigma_x \otimes {\bf 1} \otimes \sigma_x+
\sigma_y \otimes {\bf 1} \otimes \sigma_y+\sigma_z \otimes {\bf 1} \otimes \sigma_z,$$ 
$$ 
iE_3=\sigma_x \otimes \sigma_x \otimes {\bf 1} + \sigma_y \otimes \sigma_y \otimes {\bf 1} + 
 \sigma_z \otimes \sigma_z \otimes {\bf 1}, $$
$$
i E_4={\bf 1} \otimes \sigma_x \otimes \sigma_x+{\bf 1} \otimes \sigma_y \otimes \sigma_y + 
{\bf 1} \otimes \sigma_z \otimes \sigma_z,$$ 
 $$ iE_5=\sigma_x \otimes 
(\sigma_y \otimes \sigma_z-\sigma_z \otimes \sigma_y)+\sigma_y \otimes 
(\sigma_z \otimes \sigma_x-\sigma_x \otimes \sigma_z)+ \sigma_z \otimes 
(\sigma_x \otimes \sigma_y-\sigma_y \otimes \sigma_x). 
$$
An analysis of the Lie algebra ${\cal C}$ shows that it is the direct sum of one  two dimensional Abelian Lie algebra spanned by $E_1$ and $E_1+E_2+E_3$ and a three dimensional Lie algebra isomorphic to $su(2)$ and spanned by $\{E_5, (E_2-E_3), (E_2-E_4)\}$. The Lie algebra $su(2)$ has a one dimensional Cartan subalgebra which may be spanned by any non zero element. Therefore a general Cartan subalgebra ${\cal A}$ of ${\cal C}$ is such that an orthogonal basis of $i{\cal A}$ is given by $\{F_1, F_2,F_3\}$, with 
\be{F1F2F3}
F_1:=i E_1, \qquad F_2:=i(E_2+E_3+E_4), \qquad F_3:=aiE_5+bi(E_2-E_3)+ci(E_2-E_4), 
\ee 
for any, not simultaneously zero, real parameters $(a,b,c)$.

We now proceed to step $3$ of the above algorithm. The matrix $F_1$ is just the identity matrix which is diagonal in every basis. The matrix $F_2$ has eigenvalue $\lambda_1=3$ with eigenspace $Q_3$  spanned by $\{ \vec e_1, \vec e_2+ \vec e_3 +\vec e_5, \vec e_4+ \vec e_6 + \vec e_7, \vec e_8\}$ and eigenvalue $\lambda_2=-3$ with eigenspace $Q_{-3}$ spanned by $\{ \vec f_1, \vec f_2, \vec f_3, \vec f_4\}$ , with 
$\vec f_1:=\vec e_3 - \vec e_2,$ $\vec f_2:= \vec e_2 + \vec e_3 - 2 \vec e_5$, $ \vec f_3:= \vec e_6 - \vec e_4,$ $ \vec f_4:= -2 \vec e_7 + \vec e_6 +\vec e_4 $. Consider now $F_3$ acting on $Q_3$ and $Q_{-3}$. Direct verification shows that $F_3$ is zero on $Q_3$. On $Q_{-3}$ we have 
$$
F_3 \vec f_1=3c\vec f_1+ (2ia+2b+c) \vec f_2, 
$$
$$
F_3 \vec f_2=(3c+6b+6ia)\vec f_1-3c \vec f_2,  
$$
$$
F_3 \vec f_3=3b \vec f_3+(b+2c+2ia)\vec f_4,  
$$
$$
F_3 \vec f_4=(3b+6c-6ia)\vec f_3-3b \vec f_4. 
$$
which shows that the subspace $Q_{-3}$ splits into two invariant subspaces for $F_3$ 
spanned by $\{\vec f_1, \vec f_2\}$ and $\{ \vec f_3, \vec f_4\}$ respectively. Calculating the spectrum of $F_3$ on such subspaces we find that $F_3$ has eigenvalues $\pm \lambda$ on both  subspaces where 
\be{lambda}
\lambda:=2\sqrt{3c^2+3cb+3b^2+3a^2}. 
\ee
Notice that $\lambda$ is never zero otherwise we would have $a=b=c=0$ which we have excluded. Listing the eigenvalues of $F_1$, $F_2$ and $F_3$ and constructing the $M$-matrix of the above algorithm, we have 
$$
M=\begin{pmatrix} 1 & 1 &1&1&1&1&1&1 \cr 3 & 3 &3 &3 &-3&-3&-3&-3 \cr 0 & 0 & 0 & 0 & \lambda & \lambda & -\lambda & -\lambda  \end{pmatrix}. 
$$
Row reduction to transform this matrix in its Reduced Row Echelon Form which is 
$$
RREF(M)=\begin{pmatrix}  1 & 1 & 1 & 1 & 0 & 0 & 0 & 0 \cr 
0 & 0  & 0  & 0  & 1 & 1 & 0 & 0 \cr 
0  & 0  & 0  & 0  & 0 & 0 & 1 & 1
\end{pmatrix}, 
$$
corresponds to multiplication of $M$ on the left by the matrix 
$$
R:=\frac{1}{2} \begin{pmatrix} 1 & \frac{1}{3} & 0 \cr \frac{1}{2} & -\frac{1}{6} & \frac{1}{\lambda} \cr 
\frac{1}{2} & -\frac{1}{6} & -\frac{1}{\lambda}  \end{pmatrix}. 
$$
The rows of these  matrices give the coefficients for the linear combinations of $\{F_1,F_2, F_3\}$ whose eigenspaces are the sought vector spaces. In particular consider 
\be{Pi1ref}
\Pi_1:=\frac{1}{2}\left(F_1+\frac{1}{3}F_2\right), 
\ee
\be{Pi2}
\Pi_2:=\frac{1}{2}\left( \frac{1}{2} F_1-\frac{1}{6}F_2+\frac{1}{\lambda} F_3 \right), 
\ee
\be{Pi3}
\Pi_3:=\frac{1}{2}\left( \frac{1}{2} F_1-\frac{1}{6}F_2-\frac{1}{\lambda} F_3 \right).  
\ee
All these projections form a  complete set of, symmetric, mutually orthogonal idempotents (which are also 
called {\it generalized Young symmetrizers} \cite{Sym3}). In fact, one can verify directly that 
\be{completeness}
\Pi_1+\Pi_2+\Pi_3={\bf 1}, 
\ee
\be{Ortho}
\Pi_{j} \Pi_k=\delta_{j,k} \Pi_j. 
\ee
The eigenspaces corresponding to the eigenvalue $1$ of these matrices are the spaces $W$ (for $\Pi_1$), $V_1$ (for $\Pi_2$), $V_2$ (for $\Pi_3$). They coincide with the images of these matrices. 
The result is explicitly given in the following theorem. 

\begin{theorem}\label{decotheo}
Every decomposition of ${\cal H}:={\cal H}_2^{\otimes 3}=W \oplus V_1 \oplus V_2$ in invariant 
subspaces for $su^{S_3}(8)$ ($SU^{S_3}(8)$) or  $u^{S_3}(8)$ ($U^{S_3}(8)$)  
 corresponds to a triple $(a,b,c) \neq (0,0,0)$. 
An orthogonal  basis of $W$ is given by 
\be{W}
{\cal W}=\{ \vec e_1, \vec e_2+\vec e_3+\vec e_5, \vec e_4+ \vec e_6+\vec e_7, \vec e_8 \}, 
\ee
which coincides with the basis $\{\phi_0, \phi_1, \phi_2, \phi_3 \}$ in (\ref{earray}) and uniquely determines $W$ which coincides with the symmetric sector.  
An orthogonal basis of $V_1$ is given by $\{ |v_1 \rangle , |w_1\rangle \}$ with 
\be{V1}
\begin{array}{ccl}
|v_1\rangle &= & x_2|001\rangle+x_3|010\rangle+x_5|100\rangle \\
|w_1\rangle &= & x_4|011\rangle+x_6|101\rangle+x_7|110\rangle 
\end{array}
\ee
with $x_2:=\frac{5}{3} \lambda -6b-2ia-2c$, $x_3:=-\frac{\lambda}{3} + 6ia -4c+2b$, 
$x_5=-(x_2+x_3)=-\frac{4}{3} \lambda - 4ia+ 4b+6c$,  $x_4=\frac{5}{3} \lambda -6c+2ia-2b$,   
$x_6=-\frac{\lambda}{3} -6ia-4b+2c$, $
x_7=-(x_4+x_6)= -\frac{4}{3} \lambda + 4ia + 4c +6b$. 
An orthogonal  basis of $V_2$ is given by 
\be{V2}
\begin{array}{ccl}
|v_2\rangle &= & y_2|001\rangle+y_3|010\rangle+y_5|100\rangle \\
|w_2\rangle &= & y_4|011\rangle+y_6|101\rangle+y_7|110\rangle 
\end{array}
\ee
with $y_2=\frac{5 \lambda }{3}+6b+2ia+2c$, $y_3=-\frac{\lambda}{3}-6ia+4c-2b$, $y_5=-(y_2+y_3)= -\frac{4}{3} \lambda -4b-6c+4ia$, $y_4=\frac{5}{3} \lambda+6c-2ia+2b$, $y_6=- \frac{\lambda}{3}+6ia+4b-2c$, $y_7=-(y_4+y_6)=-\frac{4}{3} \lambda-4c -4ia -6b$. 
\end{theorem}
A direct computation using (\ref{lambda}) shows that the decomposition is orthogonal. 
\begin{proof}
The theorem follows by explicitly writing the matrices $\Pi_1$, $\Pi_2$ and $\Pi_3$ in (\ref{Pi1ref}), (\ref{Pi2}) (\ref{Pi3}). 

\begin{color}{black}

The matrix $\Pi_1$ is the following:
$$
\Pi_1:=\begin{pmatrix}    1& 0 & 0& 0 &0 & 0 & 0 & 0 \cr
0 & \frac{1}{3} & \frac{1}{3} & 0  & \frac{1}{3}&0 & 0 & 0   \cr 
0 & \frac{1}{3} & \frac{1}{3} & 0  & \frac{1}{3}& 0 & 0 & 0  \cr 
0 & 0 & 0 &  \frac{1}{3} & 0& \frac{1}{3} &  \frac{1}{3} &0 \cr 
0 & \frac{1}{3} & \frac{1}{3} & 0  & \frac{1}{3}   & 0 & 0 & 0 \cr 
0 & 0 & 0 &  \frac{1}{3} & 0& \frac{1}{3} &  \frac{1}{3} &0  \cr 
0 & 0 & 0 &  \frac{1}{3} & 0& \frac{1}{3} &  \frac{1}{3} &0  \cr 
0 & 0 & 0 & 0  & 0 & 0 & 0 & 1
 \end{pmatrix} 
$$
Thus the orthogonal basis for the subspace $W$, is the one given by equation (\ref{W}).

\end{color}

The work for the matrices $\Pi_2$ and $\Pi_3$ which depend on the parameters $a$, $b$, and $c$, requires some extra considerations. Let us consider  the discussion for $\Pi_2$ and $V_1$. 
The matrix $\Pi_2$ in (\ref{Pi2}) is $\Pi_2:=\frac{1}{12\lambda}(\Pi_{2,1}, \Pi_{2,2}) $ with 
{\footnotesize{
$$
\Pi_{2,1}:=\begin{pmatrix}    0 & 0 & 0& 0 \cr
0 & 4 \lambda-12b & -2 \lambda -12ia- 12c & 0   \cr 
0 & -2 \lambda + 12 i a - 12 c & 4 \lambda +12 b + 12 c & 0 
 \cr 
0 & 0 & 0 & 4 \lambda -12 c  \cr 
0 & - 2 \lambda -12 i a + 12 b +12 c & - 2 \lambda + 12 i a - 12 b & 0 \cr 
0 & 0 & 0 & -2 \lambda -12 ia -12 b 
 \cr 
0 & 0 & 0 & - 2 \lambda + 12 i a + 12 b + 12 c \cr 
0 & 0 & 0 & 0 
 \end{pmatrix} 
$$
$$
\Pi_{2,2}:=\begin{pmatrix}     0& 0 & 0& 0 \cr
 -2\lambda +12 ia + 12 b +12 c & 0 &  0 & 0 \cr 
 -2\lambda -12ia -12 b & 0 & 0 & 0 \cr 
 0 & -2 \lambda + 12 i a - 12 b & - 2 \lambda - 12 ia + 12 b+12 c  & 0 \cr 
 4 \lambda -12 c & 0 & 0 & 0 \cr 
 0 & 4 \lambda +12 b + 12 c & - 2 \lambda + 12 i a - 12 c & 0 \cr 
 0 & -2 \lambda - 12 ia -12 c & 4 \lambda -12 b & 0 \cr 
 0 & 0 & 0 & 0 
 \end{pmatrix}  
$$}}

Considering the columns $2,3,$ and $5$ of $\Pi_2$,  one sees that the sum of second, third and fifth row is zero. Therefore at the most two of these columns are linearly independent. In fact using the definition of $\lambda$ in (\ref{lambda}) it follows that only one column is linearly independent. 

Taking the first one multiplied by ${6  \lambda}$ plus the second one multiplied by $2 \lambda$  one obtains the first element in the (orthogonal) basis of ${\cal V}_1$.  This choice is made because this vector is never zero if $(a,b,c) \ne 0$. \footnote{Set all the components equal to zero. This gives $a=0$, and multiplying the  component along $\vec e_3$  by $5$ and summing the $\vec e_2$-component and the $\vec e_3$-component multiplied by $5$ one obtains $b=\frac{11}{2} c$. Plugging this in the $\vec e_2$-component (set to zero), one gets $\sqrt{199c^2}=21c$ which implies $c=0$. However this would imply $b=0$ as well, which along with $a=0$, gives a contradiction.} Analogously one obtains the second vector of the basis. In fact the form of this vector can be found by noticing that the coefficients of $\vec e_4$, $\vec e_6$ and $\vec e_7$ can be found by the ones of $\vec e_2$, $\vec e_3$ and $\vec e_5$, by exchanging 
$b \leftrightarrow c$ $a \leftrightarrow -a$.

The discussion of $\Pi_3$ and ${\cal V}_2$ is analogous. In fact, an explicit calculation shows that $\Pi_3$ can be obtained from  $\Pi_2$ with the exchanges  $a \leftrightarrow -a$, $b \leftrightarrow -b$, $c \leftrightarrow -c $. \end{proof}

This decomposition includes, as special cases, decompositions found in the standard  quantum physics literature. 
 For examples, the two dimensional space obtained with the Young symmetrizers in
 \cite{Sym3},  which can also be obtained with the recursive use of the Clebsch-Gordan coefficients \cite{Hamarmesh}, is  spanned by 
$$
\hat \psi_1:=\frac{1}{\sqrt{2}} |0 10 \rangle-  \frac{1}{\sqrt{2}} |1 00\rangle:= \frac{1}{\sqrt{2}} \vec e_3 - \frac{1}{\sqrt{2}} \vec e_5 
$$ 
$$
\hat \psi_2:=-\frac{1}{\sqrt{2}} |011\rangle  + \frac{1}{\sqrt{2}} |101\rangle   =-\frac{1}{\sqrt{2}} \vec e_4 + \frac{1}{\sqrt{2}} \vec e_6 , 
$$
and it is obtained as a special case of $V_1$ by choosing $a=0$. $b=\frac{1}{3 \sqrt{2}}$, $c=-\frac{1}{6 \sqrt{2}}$ (which give $\lambda=\frac{1}{\sqrt{2}}$).\footnote{This choice gives $\hat \psi_1$ for the first element of the basis ${\cal V}_1$ and a multiple of $\hat \psi_2$ for the second element. Notice that the elements of the bases we describe are not necessarily normalized.}

\begin{remark}\label{later}
Consider the basis $\{|v_1 \rangle, |w_1 \rangle \}$ of $V_1$ (similar consideration can be done for $V_2$) and assume $x_6=0$. Using the definition of $\lambda$ (\ref{lambda}), this gives, beside $a=0$, 
$$
\sqrt{3b^2+3c^2+3bc}=3c-6b.
$$
Squaring both terms and solving the corresponding quadratic equation we obtain that if $x_6=0$ $c=b$ and $b < 0$ or $c=\frac{11}{2} b$ and $b > 0$. The second case gives $x_7=0$ which also gives $x_4=0$ which is not possible since it would make one of the basis vectors equal to zero. Therefore not all the values of $(a,b,c)$ are possible. In the case $c=b<0$ we have besides $x_6=0$  also $x_3=0$.  
Assume now $x_7=0$. A similar reasoning as above leads to conclude that the only possibilities are $a=0$ and $c=\frac{11}{2}b$ and $c=-\frac{b}{2}$, in both cases, $b>0$. The 
case $c=\frac{11}{2}b$ is however not possible because, as we have seen, this would imply $x_6=0$.  On the other hand, the case $c=-\frac{b}{2}$ with $b>0$ also gives $x_2=0$. {\color{black} Analogous reasoning shows  that $x_4=0$ implies $x_5=0$. In fact one has $x_4=0 \leftrightarrow x_5=0$, $x_6=0 \leftrightarrow x_3=0$, $x_7=0 \leftrightarrow x_2=0$.  We shall use these considerations  in section \ref{E2S}.}  

\end{remark}

\section{Measures of Entanglement for General Three Qubits States}\label{entgen}

For a general multi-partite quantum system, a measure of 
  entanglement is a nonnegative real function on  the space of density matrices  which satisfies certain axioms. In particular it does not increase under local operations and classical communication (LOCC), it is zero on separable states (that is, statistical mixtures  of product states), it is unchanged by local unitary operations, and it is usually normalized to one (cf., e.g.,  \cite{ZK} and \cite{Hororev}   for a detailed introduction to entanglement measures). For the case of two qubits $A$ and $B$, a very common measure is the {\it concurrence} \cite{Wootters} whose square is called the 2-tangle, $\tau_{AB}$. This can be defined from the density matrix $\rho_{AB}$, by calculating the spectrum of 
  $\rho_{AB}\sigma_y \otimes \sigma_y \rho^*_{AB} \sigma_y \otimes \sigma_y $\footnote{\begin{color}{black}Given a matrix $A$ or a complex constant $c$, we always denote by $A^*$ and by $c^*$ the complex conjugate.\end{color}}
 which can be shown to be made of real and nonnegative values $\lambda_1^2 \geq \lambda_2^2 \geq \lambda_3^2 \geq \lambda_4^2$, and $\tau_{AB}$ is defined as  
  \be{tab}
\tau_{AB}=\left[ \max \{ \lambda_1-\lambda_2 -\lambda_3 -\lambda_4, 0\} \right]^2.  
\ee
Consider now a pure state where  three qubits ($A$, $B$, and $C$)  are present, which
 is the case of this paper.  One can consider after tracing out $C$, the entanglement between $A$ and $B$, $\tau_{AB}$, and analogously $\tau_{AC}$ and $\tau_{BC}$. In this case a monogamy relation holds \cite{Monog1} \cite{Monog2}: If $A$ is fully entangled with $B$, that is $\tau_{AB}=1$ then we must have $\tau_{AC}=0$, that is, the state $\rho_{AC}$ is separable
\begin{color}{black}(where $\rho_{AC}$ is the partial trace with respect to 
subsystem  $B$)\end{color}. In fact, a more refined inequality holds \cite{Coff}. Consider a pure state $\rho$ and consider the system as a bipartite system $A-(BC)$. Even though $BC$ is four dimensional, it follows from the Schmidt decomposition (see, e.g., \cite{NC} pg. 109) that only two (orthogonal) directions are necessary to express the full state. Therefore, we can treat effectively $(BC)$ as a two level system and define the entanglement $\tau_{A(BC)}$ between $A$ and $(BC)$. Then one has the following inequality which was one of the main results of \cite{Coff} 
\be{inecoff}
\tau_{AB}+ \tau_{AC} \leq \tau_{A(BC)}. 
\ee
The difference between $\tau_{A(BC)}$ and $\tau_{AB}+ \tau_{AC} $ is by definition, the {\it distributed entanglement} or $3-$tangle, which we denote simply by $\tau$, that is, the amount of entanglement not due to pairwise entanglement between the quantum bits.  Explicit formulas were given in \cite{Coff} for  $\tau_{A(BC)}$ and $\tau$. We report them below because we shall use them in our analysis. 
\begin{color}{black} Let $\rho_A$  be the partial trace with respect to 
subsystems $B$ and $C$.
\end{color}
\be{explittabc}
\tau_{A(BC)}=4\det(\rho_A); 
\ee
\be{explict}
\tau=\tau_{A(BC)}-\tau_{AB} -\tau_{AC}= 4\left|t_{000}^2t_{111}^2+t_{001}^2t_{110}^2+t_{010}^2t_{101}^2+t_{100}^2t_{011}^2-2d_1+4d_2\right|, 
\ee
with 
$$
d_1:=t_{000}t_{111}t_{011}t_{100}+t_{000}t_{111}t_{101}t_{010}+t_{000}t_{111}t_{110}t_{001}+
$$
$$
t_{011}t_{100}t_{101}t_{010}+t_{011}t_{100}t_{110}t_{001}+t_{101}t_{010}t_{110}t_{001}, 
$$
$$
d_2:=t_{000}t_{110}t_{101}t_{011}+t_{111}t_{001}t_{010}t_{100}, 
$$
  for a state 
  $$
  |\psi\rangle =\sum_{ijk}t_{ijk}|ijk\rangle, 
  $$
  with $ijk=0,1$.

\section{States in the Symmetric Sector and their Entanglement}\label{ESS}


\begin{color}{black}
For any state in the symmetric sector, because of symmetry, we have  $\rho_A=\rho_B=\rho_C$, so 
for these states, $\tau_{AB}=\tau_{AC}=\tau_{BC}$ and therefore we have, using (\ref{explittabc}) (\ref{explict}),
\end{color}
\be{TABTAC}
\tau_{AB}=\tau_{AC}=\tau_{BC}= \frac{4\det (\rho_A)-\tau}{2}. 
\ee 


In order to express the entanglement measures $\tau$ and  $\tau_{AB}=\tau_{AC}=\tau_{BC}$ in a compact 
fashion, we introduce an extra piece of notation. Define 
\be{X2X3X4}
X_2:=c_0c_2-c_1^2, \qquad X_3:=c_0 c_3-c_1 c_2, \qquad X_4:=c_1 c_3-c_2^2. 
\ee

The quantities $X_2$, $X_3$ and $X_4$ give a quick test of separability for states in the symmetric sector as described in the following proposition. 
\begin{proposition}\label{separ3} A state $\psi$  (\ref{psibas})  in the symmetric sector 
is separable if and only if 
\be{allzer}
X_2=X_3=X_4=0. 
\ee
In this case $\psi$ is a symmetric product state of the form 
\be{formasep}
\psi=\phi \otimes \phi \otimes \phi, 
\ee
with $\phi$ a one qubit state. 
\end{proposition}
\begin{proof}
Assume that $\psi$ in (\ref{psibas}) is a product state, i.e., 
$$
\psi=(\alpha_0 |0\rangle+\alpha_{1}|1\rangle) \otimes (\beta_0 |0\rangle+\beta_{1}|1\rangle) \otimes 
(\gamma_0 |0\rangle+\gamma_{1}|1\rangle). 
$$
Expanding and comparing with (\ref{psibas}), we have 
$$
c_0=\alpha_0 \beta_0 \gamma_0, 
$$
$$
c_1=\alpha_1 \beta_0 \gamma_0=\alpha_0 \beta_1 \gamma_0=\alpha_0 \beta_0 \gamma_1, 
$$
$$
c_2=\alpha_1 \beta_1 \gamma_0=\alpha_1 \beta_0 \gamma_1=\alpha_0 \beta_1 \gamma_1, 
$$
$$
c_3=\alpha_1 \beta_1 \gamma_1. 
$$
Using these in (\ref{X2X3X4}) one verifies (\ref{allzer}). For example, for $X_2$ we have 
$$
c_0 c_2=\alpha_0 \beta_0 \gamma_0\alpha_1 \beta_1 \gamma_0= (\alpha_1  \beta_0 \gamma_0)(\alpha_0\beta_1 \gamma_0)=c^2_1. 
$$
Viceversa, assume (\ref{allzer}) is verified and consider the state (\ref{psibas}). If $c_0=0$, then, 
from (\ref{allzer}), (\ref{X2X3X4}),  it follows that $c_1=0$ and $c_2=0$.  Therefore the state coincides with $|111\rangle$ which is separable and of the symmetric form (\ref{formasep}). If $c_0 \ne 0$ we can assume $c_0=1$, without loss 
of generality keeping the state not normalized. Conditions (\ref{allzer}) (\ref{X2X3X4})  give $c_2=c_1^2$, $c_3=c_1^3$. Therefore, the state $\psi$ in (\ref{psibas}) is of the form (\ref{formasep}) with $\phi=|0 \rangle+ c_1 |1\rangle$. 
\end{proof}

With the  notation (\ref{X2X3X4}),  the entanglement measures $\tau$ and $\tau_{AB}=\tau_{AC}=\tau_{BC}$ take a compact form as described in the following two propositions. 
\begin{proposition}\label{PO} The distributed entanglement  $\tau$ on the symmetric sector is given by 
\be{explitau}
\tau=4|X_3^2-4X_2 X_4|. 
\ee
\end{proposition}
\begin{proof}
Applying formula (\ref{explict}) we obtain 
\be{fromex}
\tau=4 |c_0^2 c_3^2-3 c_1^2 c_2^2-6c_0 c_1 c_2 c_3 +4c_0 c_2^3+ 4 c_3 c_1^3|. 
\ee
Direct verification using formulas (\ref{X2X3X4}) in (\ref{explitau}) shows that $\tau$ in (\ref{explitau}) coincides with (\ref{fromex}). 
\end{proof}
\begin{proposition}\label{explipar}
The pairwise entanglement $\tau_{AB}=\tau_{AC}=\tau_{BC}$ is given by 
\be{exu}
\tau_{AB}=\tau_{AC}=\tau_{BC}= 2   \left( \det(\rho_A) -|X_3^2-4X_2 X_4| \right), 
\ee
where
\be{detroA}
\det(\rho_A)=|X_3|^2+2 |X_2|^2+2 |X_4|^2. 
\ee
\end{proposition}
Therefore the expression for the pairwise entanglement is 
\be{pairW}
\tau_{AB}=\tau_{AC}=\tau_{BC}=2\left(  |X_3|^2+2 |X_2|^2+2 |X_4|^2-|X_3^2-4X_2 X_4| \right). 
\ee

\begin{proof} We explicitly write the state $\psi$ in (\ref{psibas}) as $\psi=(c_{0}, c_1, c_1, c_2, c_1, c_2, c_2, c_3)^T$ and the associated density matrix $\rho=\psi \psi^\dagger$. By taking the partial trace with respect to $B$ and $C$, we obtain, 
\be{roA}
\rho_A:=\begin{pmatrix} |c_0|^2+2 |c_1|^2+|c_2|^2 & c_0 c_1^*+2 c_1 c_2^*+ c_2 c_3^* \cr 
c_0^* c_1+2 c_1^* c_2+ c_2^* c_3 & |c_1|^2+2 |c_2|^2+|c_3|^2 \end{pmatrix}, 
\ee
and, after simplifications, 
$$
\det(\rho_A)= 2|c_0|^2 |c_2|^2+|c_0|^2 |c_3|^2 + 2 |c_1|^4 + 2 |c_1|^2 |c_3|^2+
|c_2|^2 |c_1|^2+ 2 |c_2|^4
$$
$$ - 2 c_0^* c_2^*c_1^2- c_0^*c_1 c_2 c_3^* 
-2 (c_1^*)^2c_0 c_2 - 2 c_1^*c_2^2 c_3^*- c_0 c_1^*c_2^*c_3- 2 c_1 c_3 (c_2^*)^2. 
$$
By replacing the expressions of $X_2$, $X_3$, $X_4$ in (\ref{X2X3X4}) in the right hand side of (\ref{detroA}) one verifies that it coincides with the above expression of $\det(\rho_A)$. 
\end{proof}

From (\ref{pairW}) we obtain 
$$
\frac{\tau_{AB}}{2} \geq 2 |X_2|^2+ 2 |X_4|^2-4|X_2 X_4|=2\left( |X_2|-|X_4| \right)^2 \geq 0. 
$$
To have equality, that is the pairwise entanglement equal to zero,   both inequalities   have to hold with the equal sign. We must have 
$$
|X_2|=|X_4|, \quad |X_3^2-4X_2X_4|=|X_3|^2+4|X_2||X_4|=|X_3|^2+4|X_2|^2=|X_3|^2+4|X_4|^2. 
$$
The following also follows  from Proposition \ref{separ3}. 
\bp{SSZ} The only states in the symmetric sector that have both distributed and pairwise entanglements equal to zero are the separable states. 
\ep

Pairwise entanglement $\tau_{AB}$ and distributed entanglement $\tau$ are local invariant, that is,  they are invariant under local unitary  transformations, which, in the symmetric case are taken symmetric, i.e.,  of the form $X \otimes X \otimes X$, with $X \in U(2)$. A complete set of local invariant for general three qubits states and general local unitary transformations, is known. For symmetric qubit states a complete set of invariant can be obtained using the Majorana polynomial representation of symmetric states \cite{Ribeiro}. We briefly review this.\footnote{We only discuss the Majorana polynomial representation in the three qubits case. For a general treatment, we refer to \cite{Makele} and references therein.} Given a general (not necessarily symmetric) product state 
$
\psi_1 \otimes \psi_2 \otimes \psi_3, 
$
with $\psi_j:=\alpha_j |0 \rangle + \beta_j | 1 \rangle$, $j=1,2,3$, one can obtain a symmetric state of the form (\ref{psibas}) as $A\Pi \psi_1 \otimes \psi_2 \otimes \psi_3$,  where 
\be{totalsym}
\Pi:=\frac{1}{3!} \sum_{P \in S_3} P
\ee 
is the total symmetrizer,  and $A$ is a normalization factor. In particular, direct calculation shows,  with the definitions (\ref{earray}),  
\be{uno}
\Pi (\psi_1 \otimes \psi_2 \otimes \psi_3)=  \alpha_1 \alpha_2 \alpha_3 \phi_0+ \left(\frac{\alpha_1 \alpha_2 \beta_3+ \alpha_1 \beta_2 \alpha_3 +\beta_1 \alpha_2 \alpha_3}{3}  \right)\phi_1 + 
\left( \frac{\alpha_1 \beta_2 \beta_3+ \beta_1 \alpha_2 \beta_3 + \beta_1 \beta_2 \beta_3}{3}\right) \phi_2+ 
\beta_1 \beta_2 \beta_3 \phi_3. 
\ee
Viceversa given a symmetric state (\ref{psibas}) one considers the associate Majorana polynomial
$$
p_M(x)=c_0 x^3+ 3 c_1 x^2+ 3 c_2 x+ c_3, 
$$  
which by calculating the zeros and up to a common factor can be written as 
$$
p_M(x)=( \alpha_1x+  \beta_1)( \alpha_2x+  \beta_2)( \alpha_3x+  \beta_3). 
$$ 
By choosing $\psi_j:=\alpha_j |0 \rangle + \beta_j |1 \rangle$  and using (\ref{uno}), we see that the resulting symmetric state $\Pi \psi_1 \otimes \psi_2 \otimes \psi_3$  is given by (\ref{psibas}). Therefore every symmetric state is in correspondence with a not ordered triple of one qubit states $\psi_1$, $\psi_2$, $\psi_3$. Since each qubit is in correspondence with a point on the Bloch sphere (see, e.g.,  \cite{NC}) a symmetric state is in correspondence with three not ordered vectors from the origin  to the Bloch sphere in $R^3$. Furthermore, since for $X \in U(2)$, we have 
$$
X \otimes X \otimes X \Pi (\psi_1 \otimes \psi_2 \otimes \psi_3) =\Pi (X \psi_1 \otimes X\psi_2 \otimes X\psi_3), 
$$
and applying a symmetric local unitary operation corresponds to a {\it simultaneous} rotation of the three Bloch vectors of the three one qubit states  $\psi_1$, $\psi_2$ and $\psi_3$. Therefore the {\it angles} between the Bloch vectors give a complete set of invariants under local symmetric unitary operations.

We remark here that using this representation of symmetric states it is possible to assume that the states (\ref{psibas}) can be written, after local symmetric unitary operations, in special forms. In particular, after a common 
rotation, it is possible to assume that one of the Bloch vectors corresponding to $\{ \psi_1, \psi_2, \psi_3 \}$ is in a special position, for example along the $z$-axis,  while the remaining two can be rotated arbitrarily around the the first one. One possible special form to write the state (\ref{psibas}) after a local unitary transformation is the one with \ $c_3=0$ and $c_0$ and $c_2$  real (or have the same phase, recall that states are defined up to a phase factor). In order to achieve this,  take $\psi_1 \otimes \psi_2 \otimes \psi_3$ and choose $X \in U(2)$ so that $X \psi_1=|0 \rangle$ (up to a phase factor). Therefore we have 

$$
X \otimes X \otimes X\psi_1 \otimes \psi_2 \otimes \psi_3=|0\rangle \otimes \left( \cos(\theta_1) |0 \rangle +\sin(\theta_1) e^{i \chi_1} |1\rangle \right) \otimes  \left( \cos(\theta_2) |0 \rangle +\sin(\theta_2) e^{i \chi_2} |1\rangle \right) , 
$$
for real parameters $\theta_1,$ $\theta_2$, $\chi_1$, $\chi_2$. Now we can apply 
$Y \otimes Y \otimes Y$, with $Y=\begin{pmatrix} e^{i\chi} & 0 \cr 0 & e^{{-i\chi}}  \end{pmatrix}$.  We obtain 
$$
(Y \otimes Y \otimes Y) X \otimes X \otimes X\psi_1 \otimes \psi_2 \otimes \psi_3=
$$
$$
e^{i \chi} |0\rangle \otimes \left( \cos(\theta_1) e^{i\chi} |0 \rangle +\sin(\theta_1) e^{i (\chi_1-\chi) } |1\rangle \right) \otimes  \left( \cos(\theta_2) e^{i \chi} |0 \rangle +\sin(\theta_2) e^{i (\chi_2-\chi)} |1\rangle \right) = 
$$
$$
|0\rangle \otimes \left( \cos(\theta_1)  |0 \rangle +\sin(\theta_1) e^{i (\chi_1-2\chi) } |1\rangle \right) \otimes  \left( \cos(\theta_2) |0 \rangle +\sin(\theta_2) e^{i (\chi_2-2\chi)} |1\rangle \right). 
$$
The choice $\chi:=\frac{\chi_1 + \chi_2}{4}$, gives  with $\eta = \frac{\chi_1 -\chi_2}{2}$, the form 
\be{psiprodcan}
\psi_{prodcan}:=|0\rangle \otimes \left( \cos(\theta_1)|0\rangle+ \sin(\theta_1) e^{i \eta}\right)|1\rangle \otimes \left( \cos(\theta_2)|0\rangle+ \sin(\theta_2) e^{-i \eta}|1\rangle\right). 
\ee
Applying the total symmetrizer $\Pi$ in (\ref{totalsym}) to $\psi_{prodcan}$ in (\ref{psiprodcan}), one obtains a symmetric state (\ref{psibas}) with $c_3=0$ and $c_0$ and $c_2$ real. {\color{black} We remark that $3$  is the minimum number of parameters necessary to identify equivalence classes of (unitary) locally equivalent states since states  are identified up to a common phase factor and therefore (in the symmetric sector) by $6$ parameters  and $SU(2)$ has dimension $3$.}

\section{States in the Two Dimensional Invariant Subspaces and their Entanglement}\label{E2S}

We now consider the  invariant subspaces of dimension two: $V_1$ and $V_2$ described in section \ref{Para2}. Since the orthogonal  basis of $V_2$, given in equation (\ref{V2}), can be obtained from  the orthogonal  basis of $V_1$, given in equation (\ref{V1}), exchanging   $a \leftrightarrow -a$, $b \leftrightarrow -b$, $c \leftrightarrow -c $, and $(a,b,c)$ are free parameters (not all zero) we will consider without loss of generality only the subspace $V_1$. We shall calculate the pairwise entanglements and the distributed entanglement. We remark that since these states are in general not invariant under permutation (as opposed to states in the symmetric sector treated in the previous section) there is no a priori reason why $\tau_{AB}$ should be equal to $\tau_{AC}$.

In order to simplify the calculation of the entanglement measures, it is convenient to anticipate a result on  dynamics (treatment of  dynamics will be done in the next section). 
We recall that we call {\it local (special) unitary symmetric operations}, operations of the type $X \otimes X \otimes X$ with $X \in U(2)$ ($X \in SU(2)$). 

\bp{locali}
Given two states $|\psi_1\rangle$, $|\psi_2\rangle$ in the subspaces $V_1$ (same for $V_2$) it is always possible to go from  $|\psi_1\rangle$ to  $|\psi_2\rangle$, using local operations.
\ep
\bpr The Lie algebra corresponding to the Lie group of local symmetric special unitary matrices  is spanned by the matrices $iH_{x,y,z}$  defined in (\ref{Hxy}). This Lie algebra leaves $V_1$ invariant. It is in fact the standard representation of $su(2)$. This can be verified directly. Explicit computation using (\ref{V1}) shows that 
$H_z |v_1 \rangle= |v_1 \rangle $,  $ H_z |w_1\rangle=- |w_1 \rangle $. 
 Calculating $iH_x|v_1\rangle$  we get  
$$
iH_x |v_1 \rangle =-i x_3 |101\rangle -i x_2 |110 \rangle + i (x_2+ x_3) |011\rangle. 
$$
This is proportional to $|w_1\rangle$ if there exists a non zero $k$ such that $x_3=kx_6$ and $x_2=kx_7$. Now we apply the considerations of the remark \ref{later}. If $x_6=0$ then we cannot have $x_7=0$ but we have $x_3=0$. Therefore we choose $k=\frac{x_2}{x_7}\not=0$. If $x_7=0$, we cannot have $x_6=0$ but we have $x_2=0$ and we choose $k=\frac{x_3}{x_6} \not=0$. If both $x_6$ and $x_7$ are different from zero a direct calculation using the definitions of $x_{2,3,6,7}$ shows that $\frac{x_2}{x_7}=\frac{x_3}{x_6}$. Therefore, defining 
$k:=\frac{x_2}{x_7}=\frac{x_3}{x_6}\not=0$, we have the result. 
\begin{color}{black}
In all cases,  we have $$H_x |v_1 \rangle =- k |w_1 \rangle.$$ 
It can be verified, using the fact that $x_2+x_3+x_5=0$, that $H_x^2 |v_1\rangle = |v_1 \rangle$. Thus applying $H_x$ to the previous equation  $H_x |w_1\rangle=- \frac{1}{k} |v_1 \rangle$. Therefore we have $H_x k |w_1 \rangle=-|v_1\rangle$, $H_x |v_1 \rangle=-k |w_1  \rangle$. 
\end{color}

Thus on the orthogonal basis $\{ v_1, k |w_1\rangle \}$ $H_z$ and $-H_x$ act as $\sigma_z$ and $\sigma_x$ on the basis $\{|0 \rangle, |1 \rangle\}$\footnote{Notice also that because of how $k$ is defined the vectors $|v_1\rangle$ and $|w_1\rangle$ have the same length and therefore can be normalized simultaneously.} and therefore $iH_z$ and $-iH_x$ generate the Lie algebra $su(2)$. Since the corresponding Lie group, $SU(2)$ is  transitive on the complex sphere, the symmetric local transformations $X \otimes X \otimes X$, $X \in SU(2)$, are able to transfer any state to any other state (cf. \cite{JS}). 

\epr

We will now compute the entanglement measures for the states in the invariant subspaces of dimension $2$.   Since these quantities remain unchanged by using local operations, and  all the states in $V_1$ (or $V_2$) can be reached using local operation starting from an arbitrary state, as proved in the previous proposition, it is enough to   compute the measures  for a particular state.

First we will see that the distributed entanglement $\tau$ is always zero.

\bp{tanglezero}
Let $|\psi\rangle\in V_1$ (or $|\psi\rangle\in V_2$) then $\tau=0$
\ep

\bpr
Direct calculation  using (\ref{explict})  shows that $\tau=0$ for the first basis vector $|v\rangle =  x_2|001\rangle+x_3|010\rangle+x_5|100\rangle $ in the definition (\ref{V1}). Therefore the result follows using Proposition \ref{locali}.

\epr

To calculate  $\tau_{A(BC)}$, $\tau_{AB}$ and $\tau_{AC}$, we 
 recall that since  $\tau=0$ from (\ref{explict})
 \be{fol}
\tau_{A(BC)}=\tau_{AB}+\tau_{AC}.  
 \ee
 Since these quantities are constant on $V_1$, let us calculate them at $|v_1 \rangle$ in (\ref{V1}). We only need to compute $\tau_{A(BC)}$ and $\tau_{AB}$ since $\tau_{AC}$ will follow from (\ref{fol}). Computation of  $\rho_A$ and $\rho_{AB}$ for the state $|v_1 \rangle$ gives

We have:
\be{rhoA}
\rho_A=\left( \begin{array}{cc}
 |x_2|^2 + |x_3|^2  &     0 \\
0 &  |x_5|^2
 \end{array}
 \right), 
\ee

\be{rhoAB}
\rho_{AB}=\left( \begin{array}{cccc}
 |x_2|^2  & 0   &  0  & 0 \\
 0   &|x_3|^2  & x_3{x}^*_5 & 0 \\
 0  &  x_5   x^*_3  & |x_5|^2  & 0 \\
 0&0 &0  &  0 
 \end{array}
 \right).
\ee
Using (\ref{explittabc}), we obtain

\be{tauABC}
\tau_{A(BC)}= 4 \left(|x_2|^2+|x_3|^2\right)|x_5|^2= 4(|x_2|^2 + |x_3|^2) |x_2 +x_3|^2. 
\ee

To compute $\tau_{AB}$ one has to calculate the eigenvalues of $\rho_{AB} \sigma_y \otimes \sigma_y  \rho^*_{AB} \sigma_y \otimes \sigma_y$ which using formula (\ref{rhoAB}) can be seen to be:  zero with multiplicity two and the eigenvalues of the $2 \times 2$ matrix
$$
\begin{pmatrix} 2 |x_3|^2 |x_5|^2 & 2|x_3|^2 Re(x_5   x^*_3) \cr 
2|x_5|^2Re(x_5   x^*_3) & 2  |x_5|^2 |x_3|^2 \end{pmatrix},  
$$ 
which are 
$\lambda_1^2=2|\mu|(|\mu|+ |Re(\mu)|)$, $\lambda_2^2=2|\mu|(|\mu|-| Re(\mu)|)$, with $\mu:=x_5 x_3^*$. 
 Using  formula (\ref{tab}) we have since $\lambda_3=\lambda_4=0$, 
 \be{simpl99}
\tau_{AB}=(\lambda_1 - \lambda_2)^2= 2 |\mu| \left( \sqrt{|\mu|+ |Re(\mu)|} -\sqrt{|\mu|-  |Re(\mu)|}  \right)^2= 
4 |\mu|\left( |\mu|- |Im(\mu)| \right). 
 \ee
We also have $\tau_{AC}=\tau_{A(BC)}-\tau_{AB}$. 

{ \color{black}  
\br{productstates} From formula (\ref{tauABC})  it follows that $\tau_{A(BC)}$ is zero if and only if $x_5=0$ ($x_2$ and $x_3$ cannot be simultaneously zero because this would imply the vector $|v_1\rangle$ to be zero). In this case,  $|v_1\rangle$ would be a product state (of the form $|0\rangle \otimes \tilde \psi_{BC}$ for a state $\tilde \psi_{BC}$ on the subsystem $B-C$). Since the local symmetric unitary group is transitive on the subspace $V_1$ every state in this subspace is a product state as expected when the entanglement is zero. The condition on the entanglement $\tau_{AB}$ is less intuitive. It can be state by saying that $\mu:=x_5x^*_3$ is purely imaginary.  
\er
}

\section{Symmetric Dynamics on the Invariant Subspaces}\label{SSD1}

We now study how the Lie group $U^{S_3}(8)$ of symmetric dynamics, i.e., unitary transformations which commute with the permutation group of three objects, $S_3$,  acts on its invariant subspaces. On the two dimensional invariant subspaces $V_1$ and $V_2$ this group acts as $U(2)$ and its induced dynamics is not more rich than the one of the group  generated by  the symmetric local transformations 
$X \otimes X \otimes X$ and multiples of the $8 \times 8$ identity. These transformations do not modify the entanglement measures and in particular the distributed entanglement which is zero for any subspace as we have seen in  the previous section.

 More interesting is the dynamics of $U^{S_3}(8)$ on the four dimansional symmetric sector which can be proven (see, e.g., \cite{noisym1}) to be given by all possible unitary $4\times 4$ matrices. The local symmetric unitary transformations are a Lie subgroup of $U^{S_3}(8)$ whose Lie algebra is spanned by $iH_{x,y,z}$ and multiples of the identity. If we consider the orthonormal basis $\{ \phi_0, \frac{\phi_1}{\sqrt{3}}, \frac{\phi_2}{\sqrt{3}}, \phi_3\}$, the matrices $\frac{i}{2}H_{x,y,z}$ give the four dimensional irreducible representation of $su(2)$. They can be computed using the Clebsch-Gordan coefficients \cite{Hamarmesh} (or directly by computing their action on the given basis).   They are 
\be{Sx}
S_x:=\frac{1}{2}\begin{pmatrix}      0 & \sqrt{3} i  & 0 & 0 \cr  \sqrt{3} i & 0 &  2 i & 0 \cr 
0 & 2 i & 0 & \sqrt{3}i  \cr 0 & 0 &  \sqrt{3} i  & 0 \end{pmatrix}, 
\ee   
\be{Sy}
S_y:=\frac{1}{2}\begin{pmatrix}  0 & -\sqrt{3}& 0 & 0 \cr \sqrt{3} & 0 & -2 & 0 \cr 
0 & 2 & 0 & -\sqrt{3} \cr 0 & 0 & \sqrt{3} & 0 \end{pmatrix}, 
\ee
\be{Sz}
S_z:=\frac{1}{2}\begin{pmatrix} 3 i & 0 & 0 & 0 \cr 
0 & i & 0 & 0 \cr 
0 & 0 & -i & 0 \cr 
0 & 0 & 0 & -3i \end{pmatrix},  
\ee
and  satisfy the standard commutation relations for $su(2)$, 
\be{commurel2}
[S_x, S_y]=S_z, \qquad [S_y, S_z]=S_x, \qquad [S_z, S_x]=S_y. 
\ee

The Lie subgroup of $SU(4)$ corresponding to the Lie algebra spanned by $S_{x,y,z}$ ( corresponding to local symmetric transformations)  leaves the measures of entanglement unchanged. In fact, it is a {\it maximal} Lie subgroup  
having this property   as shown in the following two propositions whose proofs are postponed to the Appendix. 
\begin{proposition}\label{onlyone1}
The local Lie group corresponding to the Lie algebra spanned by the $4 \times 4$ $i \times$ identity and $S_{x,y,z}$ is maximal among the  Lie groups  leaving the distributed tangle $\tau$ unchanged on the $4$-dimensional symmetric sector. That is, there is no Lie group leaving such measure invariant which is larger than the local Lie group. 
\end{proposition}

\begin{proposition}\label{onlyone2}
The local Lie group corresponding to the Lie algebra spanned by the $4 \times 4$ identity and $S_{x,y,z}$ is a maximal Lie group  leaving the pairwise tangle $\tau_{AB}=\tau_{AC}=\tau_{BC}$ unchanged on the $4$-dimensional symmetric sector. 
\end{proposition}

The symmetric Lie group $U^{S_3}(8)$ acts on the symmetric sector spanned by the vectors $\{\phi_0, \phi_1, \phi_2, \phi_3\}$ as the unitary group $U(4)$ and the associated Lie algebra acts like $u(4)$ (see, e.g., \cite{noisym1}). The Lie algebra spanned by $S_{x,y,z}$ in (\ref{Sx}), (\ref{Sy}), (\ref{Sz}) is a subalgebra of $u(4)$ corresponding to local symmetric operations and isomorphic to $su(2)$, the symmetric sector giving an irreducible representation of $su(2)$. Since $U(4)$ is transitive on the complex unit sphere, it is possible,  using elements of the Lie group $U^{S_3}(8)$,  to transfer the state from a product state $\phi \otimes \phi \otimes \phi$ to any arbitrary state in the  symmetric sector,  independently of the entanglement value of the target state. In order to analyze how the elements of $U^{S_3}(8)$ modify  the entanglement on the symmetric sector,  we analyze the structure of the Lie algebra $u(4)$ starting with describing how  
 the Lie algebra spanned by $\{S_x, S_y, S_z\}$ which leaves such measures unchanged `sits' in $u(4)$.   {\color{black} Our goal is to arrive at a factorization of elements of $U(4)$ which separates  factors which modify the entanglement measures from the symmetric local transformations that do not, trying to maximize the latter ones.}

The Lie algebra $u(4)$ admits, up to conjugacy, a Cartan decomposition (see, e.g., \cite{Helgason}) 
\be{Cartandec}
u(4)=sp(2) \oplus sp^\perp(2), 
\ee
where $sp(2)$ is the symplectic Lie algebra and  $sp^\perp(2)$ is  its orthogonal complement.\footnote{The inner product considered is the inner product $\langle A, B \rangle= kTr (AB^\dagger)$, for an appropriate positive constant $k$.} They satisfy the basic (Cartan-like) commutation relations 
\be{commurelCart}
\left[ sp(2), sp(2) \right] \subseteq sp(2), \qquad \left[ sp^\perp(2), sp(2) \right] \subseteq sp^\perp(2), \qquad \left[ sp^\perp(2), sp^\perp(2) \right] \subseteq sp(2). 
\ee 
As it is customary,  we denote by $Sp(n)$ the connected Lie group associated with $sp(n)$. 
According to Cartan decomposition theorem, every unitary $4 \times 4$ matrix $U$  can be written as 
\be{Cartandeco}
U:=K_1 A K_2, 
\ee
where $K_{1,2}$ are in $Sp(2)$ and $A$ is  the exponential of an element in a  maximal Abelian subalgebra in $sp^\perp(2)$  which in this case has dimension $2$ since we are including multiples of the identity as well in $sp^\perp(2)$.\footnote{In general the maximal Abelian subalgebra for the Cartan decomposition $su(n)=sp(\frac{n}{2})\oplus sp^\perp (\frac{n}{2})$ has dimension $\frac{n}{2}-1$ which would give $1$ in this case. However, we have included multiples of the identity since we looked at the decomposition for $u(4)$. A full treatment of the decompositions for $u(n)$ can be found in \cite{Helgason} and a summary with applications for quantum systems can be found in \cite{Mikobook}.}

The symplectic Lie algebra $sp(n)$ (which has dimension $n(2n+1)$) and its associated Lie group $Sp(n)$ have several important properties that are of interest for the study of quantum dynamics. In particular, $sp(n)$ is a {\it maximal Lie subalgebra} of $su(2n)$ which means that $sp(n)$ along with any nonzero element $X \notin sp(n)$ of $su(2n)$ generates  all of $su(2n)$. Every Lie subalgebra of $su(2n)$ which is isomorphic to $sp(n)$ is actually {\it conjugate} to $sp(n)$. Furthermore, $Sp(n)$ is {\it transitive} on the complex unit sphere $S^{2n-1}$ representing quantum states. 
This means that, for any two (normalized) quantum states $|\psi_1\rangle$   and $|\psi_2\rangle$, there exists a matrix in $Sp(n)$ such that $|\psi_2 \rangle =X |\psi_1\rangle$. This means,  in particular, that any possible value of the entanglement in the symmetric sector can be achieved by only using the transformations $K_1$ and $K_2$ in (\ref{Cartandeco}).


For our purposes we consider a  Lie subalgebra ${\cal S}$ {\it conjugate}   to $sp(2)$ in $u(4)$. We consider  the Lie algebra of $4 \times 4$ matrices of the form 
$$
F:=\begin{pmatrix} ir & \alpha & \beta & \gamma \cr 
- \alpha^* & is & \delta& -\beta \cr 
-  \beta^* & -  \delta^* & -is & \alpha \cr 
-   \gamma^* &   \beta^* & -  \alpha^* & -ir \end{pmatrix}, 
$$
with $r$ and $s$ arbitrary real numbers and $\alpha,$ $\beta$, $\gamma$, $\delta$ arbitrary complex numbers. Matrices in this Lie algebra ${\cal S}$  satisfy, 
$$
FJ+JF^T=0, 
$$ 
where $J$ is the matrix 
\be{J}
J:=\begin{pmatrix} 0 & 0 & 0 & 1 \cr 0 & 0 & -1 & 0 \cr 0 & 1 & 0 & 0 \cr -1 & 0 & 0 & 0  \end{pmatrix}, 
\ee
and formulas (\ref{Cartandec})-(\ref{Cartandeco}) hold with $sp(2)$ replaced by ${\cal S}$. 
The reason for this choice is that the matrices 
$S_{x,y,z}$ (\ref{Sx}), (\ref{Sy}), (\ref{Sz}), giving the $4$-dimensional irreducible representation of $su(2)$,   belong to ${\cal S}$ (see, e.g., \cite{Dynkin}  for a treatment of how irreducible representations of $su(2)$ fit in the corresponding unitary Lie algebra). The decomposition (\ref{Cartandeco}) therefore holds with $K_1$ and $K_2$ belonging to the connected Lie group $e^{\cal S}$ conjugate to $Sp(2)$ and associated with the Lie algebra ${\cal S}. $\footnote{Here and in the following we use the convention of denoting by $e^{\cal S}$ the connected Lie group associated with a Lie algebra ${\cal S}$.} On the other hand, the for the matrix $A$ in (\ref{Cartandeco}), we can take the product of the exponentials of two elements in a Cartan subalgebra in 
${\cal S}^\perp$.   For such a  Cartan subalgebra,  we take $\texttt{span} \{ i {\bf 1}_4, i H_{zz}\}$ where on $(C^2)^{\otimes 3}$, $H_{zz}$ is defined as $H_{zz}$ in (\ref{Hzz}). 
In the symmetric sector, in the basis $\{\phi_0, \phi_1, \phi_2, \phi_3\}$ it is given by the matrix 
\be{MatriceHzz}
iH_{zz}=\begin{pmatrix} 3i & 0 & 0 & 0 \cr 0 & -i & 0 & 0 \cr 0 & 0 & -i & 0 \cr 0 & 0 & 0 & 3i \end{pmatrix}. 
\ee
With this choice formula (\ref{Cartandeco}) can be written as 
\be{Cartandeco2}
U=K_1 e^{i{\bf 1}_4 z} e^{i H_{zz} w} K_2, 
\ee
for real parameters $z$ and $w$, with $K_1$  and $K_2$ in $e^{\cal S}$ (isomorphic to $Sp(2)$). We now turn to the factorization of $K_1$ and $K_2$ in (\ref{Cartandeco2}).

Even $sp(2)$ has a Cartan decomposition $sp(2)={\cal L} \oplus {\cal L}^\perp$ that can be chosen (up to conjugacy) between two possibilities denoted by {\bf CI} and {\bf CII} (cf. Chapter X in \cite{Helgason}). Given such a decomposition, the matrices $K_1$ and $K_2$ in (\ref{Cartandeco})  
can be written as $K$ in the following formula 
\be{K12}
K=L_1 \hat A L_2, 
\ee
where $\hat A$ is the exponential of a matrix in  a maximal Abelian subalgebra inside ${\cal L}^\perp$. We choose the decomposition {\bf CI} because this allows us to separate $S_x$, $S_y$ and $S_z$ in 
${\cal L}$ and ${\cal L}^\perp$. In particular we have ${\cal L}={\cal S} \cap so(4)$ which is given by the matrices of the form 
\be{hatL}
\hat L:=\begin{pmatrix} 0 & a & k & r \cr -a & 0 & f & -k \cr -k & -f & 0 & a \cr -r & k & -a & 0 \end{pmatrix}, 
\ee
for four real parameter $a,k,r,f$.\footnote{The Lie algebra ${\cal S}$ is only conjugate to $sp(2)$, therefore, the Lie algebra characterizing the decomposition {\bf CI} is different but conjugate with respect to the one for  $sp(2)$. The matrix $M$ which gives the conjugacy  is $M:=\begin{pmatrix}  U_1 & 0 \cr 0 & U_1^T\end{pmatrix}$ with $U_1=\frac{1}{\sqrt{2}}\begin{pmatrix} 1 & 1 \cr -1 & 1 \end{pmatrix}$. We have $Msp(2)M^T={\cal S}$.} The matrix $S_y$ in (\ref{Sy}) belongs to the Lie subalgebra ${\cal L}$ (choose $a=-\frac{\sqrt{3}}{2}$, $f=-1$ and the other parameters equal to zero) while $S_{x}$ and $S_z$ belong to 
${\cal L}^\perp$.   The dimension of the Cartan subalgebra associated to this decomposition (the {\it rank} of the decomposition) is $2$. We choose as Cartan subalgebra the one spanned by $S_z$ and $H$, where $H=\texttt{diag}(0, \frac{i}{2}, -\frac{i}{2}, 0)$. Therefore $\hat A$ in formula (\ref{K12})   can be written as 
\be{Acappello}
\hat A=e^{S_z x} e^{Hy}, 
\ee
for real values $x$ and $y$. The Lie algebra ${\cal L}$ is isomorphic to $u(2)$ with the isomorphism given by 
\be{isoid}
i{\bf 1}_2 \leftrightarrow J, 
\ee
with $J$ in (\ref{J}), 
\be{isosx}
\frac{i}{2} \begin{pmatrix} 0 & 1 \cr 1 & 0  \end{pmatrix} \leftrightarrow \frac{1}{2} \begin{pmatrix} 0 & 0  & 0  & 1 \cr 0 & 0 & 1 & 0 \cr 0 & -1 & 0 & 0 \cr -1 & 0 & 0 & 0  \end{pmatrix}  
\ee
\be{isosy}
\frac{1}{2} \begin{pmatrix} 0 & 1 \cr -1 & 0  \end{pmatrix} \leftrightarrow \frac{1}{2} \begin{pmatrix} 0 & 1  & 0  & 0 \cr -1 & 0 & 0 & 0 \cr 0 & 0  & 0 & 1 \cr 0 & 0 & -1 & 0  \end{pmatrix},  
\ee
\be{isosz}
\frac{i}{2} \begin{pmatrix} 1 & 0 \cr 0 & -1  \end{pmatrix} \leftrightarrow \frac{1}{2} \begin{pmatrix} 0 & 0  & 1  & 0 \cr 0 & 0 & 0 & -1 \cr -1 & 0  & 0 & 0 \cr 0 & 1 & 0 & 0  \end{pmatrix}.   
\ee
Consider the matrix in ${\cal L}$ 
\be{erre}
R:=\begin{pmatrix} 0 & 1 & 0 & 0 \cr -1 & 0 & -\sqrt{3} & 0 \cr 0 & \sqrt{3} & 0 & 1 \cr 0 & 0 & -1 & 0  \end{pmatrix}, 
\ee
which is orthogonal to $S_y$. Performing an Euler-like decomposition on $e^{\cal L}$, we can write any element in $e^{\cal L}$, such as $L_1$ and $L_2$ in (\ref{K12}),  as 
$L:=e ^{S_y t_1} e^{Rt_2} e^{S_y t_3} e^{J t_4}$. Combining this with  $\hat A$ in (\ref{Acappello}) we obtain that every element $K$  in $e^{\cal S}$ can be written as 
\be{finalK}
K:=e ^{S_y t_1} e^{Rt_2} e^{S_y t_3} e^{J t_4} e^{S_z t_5} e^{Ht_6}e^{J t_7}e ^{S_y t_8} e^{R t_9} e^{S_y t_{10}}. 
\ee
In this decomposition, the only elements that change the entanglement measures are factors of the type $e^{Rt}$, 
$e^{J t}$, and $e^{Ht}$.  In the resulting decomposition for $U \in U(4)$ (\ref{Cartandeco2})  to these needs  to be added the evolution $e^{i H_{zz} w}$ in (\ref{Cartandeco2}). The full decomposition of unitary transformations which combines (\ref{Cartandeco2}) with (\ref{finalK}) separates factors which do not change the entanglement on the symmetric sector from factors that do. 

\section{Control of a Three Spin Symmetric Network}\label{App}

In this section, we briefly consider a possible physical scenario where the analysis carried out in the previous sections applies. We consider in particular a network of three identical  spin $\frac{1}{2}$ particles subject to a  controlling electromagnetic field which simultaneously interacts with the three spins. The three spins  interact with each other via Ising interaction. The time varying Hamiltonian for this system is 
\be{Hamilt}
H_S:=H_S(t)=H_{zz}+H_xu_x(t)+H_yu_y(t)+H_z u_z(t), 
\ee   
with $H_{zz}$ defined in (\ref{Hzz}) and $H_{x,y,z}$ defined in (\ref{Hxy}). The functions $u_{x,y,z}=u_{x,y,z}(t)$ represent (spatially uniform) components of an electromagnetic field in the $x,y,z$ direction. The dynamics of the state is given by the Schr\"odinger equation
\be{Scroeq}
\dot \psi=-iH_S(t) \psi, \qquad \psi(0)=\psi_0. 
\ee
The controllability analysis of a quantum system (see, e.g., \cite{Mikobook}) describes the set of states that can be reached starting from $\psi_0$. In this case, such a set is \cite{noisym1}
\be{reachstat}
{\cal R}_{\psi_0}:= \{ X \psi_0 \, | \, X \in U^{S_3}(8) \}. 
\ee
In particular, if $\psi_0$ belongs to one of  the invariant subspaces of $ U^{S_3}(8) $ (or $u^{S_3}(8)$),  then, for every control function, the state  will remain  in that subspace. In fact, every state in that subspace can be reached from $\psi_0$ with an appropriate control (transitivity of the unitary group \cite{noisym1}). The invariant subspaces for the group $U^{S_3}(8)$ were described in section \ref{Para2}. 

Let us restrict ourselves to the symmetric sector where the dynamics was  described in section \ref{SSD1} and let us pose the problem of reaching from a symmetric product state a state with maximum distributed entanglement $\tau=1$, using an appropriate control function. The problem will be solved if we prove that there exists a symmetric local state $\hat \psi_0:=\phi \otimes \phi \otimes \phi$ and a time $\bar t$ such that $e^{-i H_{zz} \bar t} \hat \psi_0$ has maximum distributed entanglement. This property is referred to, in the case of (pairwise) entanglement of two qubits, as  $e^{-i H_{zz} \bar t}$ being a {\it perfect entangler} \cite{ZS} and we shall adopt this terminology here, mutatis mutandis. If there exists such a product state $\hat \psi_0$, we can use in (\ref{Scroeq}) high amplitude short time pulses (so that we can neglect in the dynamics the effect of $H_{zz}$) which will produce a local symmetric transformation from $\psi_0$ to the state $\hat \psi_0$. Then we can set the controls equal to zero and allow the evolution go as $\hat \psi_0 \rightarrow e^{-iH_{zz} t} \hat \psi_0$ for time $\bar t$. We are left with proving that $e^{-iH_{zz} \bar t}$ is a perfect entangler. 

\bp{PE} There exists a $\bar t$ such that $e^{-iH_{zz} \bar t}$ is a perfect (distributed) entangler transferring the local symmetric state 
\be{chosen}
\hat \psi_0= \left( \frac{1}{\sqrt{2}} |0 \rangle + \frac{1}{\sqrt{2}} |1 \rangle \right)^{\otimes 3}, 
\ee
to a state of maximum distributed entanglement $\tau=1$. 
\ep 
\bpr The state $\hat \psi_0$ in (\ref{chosen}) can be written in the form  (\ref{psibas}) with $c_0=c_1=c_2=c_3=\frac{1}{2\sqrt{2}}$. Using  the explicit expression of $H_{zz}$ in (\ref{MatriceHzz}) we have that $c_{0,1,2,3}$ vary with time as 
\be{c0123}
c_0(t)=c_3(t)=\frac{e^{3it}}{2 \sqrt{2}}, \qquad c_1(t)=c_2(t)=\frac{e^{-it}}{2 \sqrt{2}}. 
\ee 
Using the formula for the distributed entanglement $\tau$ in (\ref{explitau}) together 
with the expressions of $X_2$, $X_3$, $X_4$ in (\ref{X2X3X4}) we get $X_2=X_4=\frac{1}{8}(e^{2it}-e^{-2it})$, 
$X_3=\frac{1}{8} (e^{6it}-e^{-2it})$, and 
$$
\tau(t)=4\left| X_3^2-4X_2 X_4\right|=\frac{1}{16} \left| (e^{6it} - e^{-2it})^2-4(e^{2it} - e^{-2it})^2 \right|=
\frac{1}{16} \left| (e^{6it} -e^{-2it})^2+ 16 \sin^2(2t)  \right|. 
$$
The function on the right hand side has a maximum equal to $1$ at $t=\frac{\pi}{4}$. Therefore the proposition holds with $\bar t=\frac{\pi}{4}$. 
\epr

\section{Conclusions}\label{Summa}

In this paper,  we have given an analysis of the states of a 
three qubit quantum system under the action of the Lie group $U^{S_3}(8)$ of unitary matrices which commute with the symmetric group of three objects. This is motivated by the controlled dynamics 
of symmetric spin networks with  three 
spin $\frac{1}{2}$ particles, as described in section \ref{App}. The Hilbert space of three qubits  splits into subspaces  of dimension $4$, $2$ and $2$, which are invariant  under the action of the Lie group $U^{S_3}(8)$. The subspace of dimension $4$ is uniquely determined and corresponds to the so called symmetric sector $W$ of states which are invariant under permutation  (symmetric states). The subspaces of dimension $2$ are not uniquely determined although they are    orthogonal to $W$ and orthogonal to each other. We have provided the following results: 

\begin{itemize}
\item[] We have parametrized all the possible decompositions of the state space in 
 invariant subspaces under the Lie group $U^{S_3}(8)$ (Theorem \ref{decotheo}). 
 
 \item[] For states in the symmetric sector $W$, we have introduced three quantities $X_2,$ $X_3$, $X_4$ in (\ref{X2X3X4}) which are easily calculated from the expression of the state and give a simple criterion of separability (Proposition \ref{separ3}). 
 
 \item[] We have calculated and expressions of the distributed entanglement and the pairwise entanglement in terms of the quantities $X_2$, $X_3$, $X_4$ Propositions \ref{PO} , \ref{explipar} and concluded that the only states on the symmetric sector which have both entanglements equal to zero are the separable states (Proposition \ref{SSZ}). 
 
 \item[] For states in the two dimensional invariant subspaces, we have proven that the distributed entanglement is always equal to zero (Proposition \ref{tanglezero}) while the pairwise entanglement will depend on the subspace considered. We provided a simple formula for it (formula (\ref{simpl99})).  
 
 \item[] We have proven that it does not exists any connected Lie subgroup of $U^{S_3}(8)$ which properly contains the Lie subgroup of local symmetric transformations and leaves unchanged the distributed entanglement (Proposition \ref{onlyone1}) or the pairwise entanglement (Proposition \ref{onlyone2}) on the symmetric sector.  
 
 \item[] We have given a decomposition of any evolution in $U^{S_3}(8)$ on the symmetric sector into (local) elements which do not modify the entanglement and factors which modify it (formulas (\ref{Cartandeco2}) and (\ref{finalK})). 
 
 \item[] We have proven that the free evolution given by a pairwise Ising interaction gives a perfect entangler for distributed entanglement on a symmetric network of three spin (Proposition \ref{PE}) and used this to prescribe a control law to drive a separable symmetric state to a state of maximal distributed entanglement.   

\end{itemize}
{\color{black}
In future research, it will be of interest to extend the results presented here to symmetric states for more than three qubits. Such extensions however will depend on a better understanding of multipartite entanglement beyond the three qubits case.}

 \section{Acknowledgement} D. D'Alessandro's  research was supported by NSF under Grant EECS-1710558
 
\section*{Data Availability Statement} The data that support this study are available from the corresponding author upon reasonable request.

\section*{Appendix: Proofs of Propositions \ref{onlyone1} and \ref{onlyone2}}

We consider a general vector in the symmetric sector 
$$
\psi=c_0 \phi_0 + \hat c_1 \frac{1}{\sqrt{3}} {\phi_1}+  \hat c_2 \frac{1}{\sqrt{3}} {\phi_2}+c_3 \phi_3, 
$$
where, comparing  with (\ref{psibas}), we have $\hat c_1=\sqrt{3}c_1$,  $\hat c_2=\sqrt{3}c_2$. 


We also follow the convention of denoting by $y_R$ and $y_I$ the real and imaginary part of a quantity $y$.

\subsection*{Proof of Proposition \ref{onlyone1}}

\begin{proof} 
\begin{color}{black}
Recall, see (\ref{explitau}), that the distributed tangle $\tau$ is given by $\tau=4|X_3^2-4X_2 X_4|$.
Consider  the quantity $X_3^2-4X_2 X_4$  written separating its real and imaginary parts as $X_3^2-4X_2 X_4:=R+iI$
\end{color}
Then invariance of $\tau$ is equivalent to invariance of the function $f:=\frac{1}{2}(R^2+I^2)$. $R$ and $I$ are functions of the complex vector $\vec v:=(c_0, \hat c_1, \hat c_2, c_3)^T:=\vec v_R+i \vec v_I$, for real vectors $\vec v_R$ and $\vec v_I$,  where $\vec v_R:= (c_{0,R}, \hat c_{1,R}, \hat c_{2,R}, c_{3,R})^T$,  $\vec v_I:= (c_{0,I}, \hat c_{1,I}, \hat c_{2,I}, c_{3,I})^T$. If $F$ is an element of the Lie algebra associated to a given Lie subgroup of $U(4)$, under the action of $e^{Ft}$, $\vec v$ changes as $\vec  v \rightarrow e^{Ft} \vec v$ and therefore $f(t)$ varies as 
\be{howitvar}
f(t)=\frac{1}{2}\left(R^2(e^{Ft} \vec v) + I^2( e^{Ft}\vec v) \right). 
\ee
Since $F$  is skew-Hermitian, it can be written as $F=A+iB$ with $A$ skew symmetric and $B$ symmetric (with both of them real).  Invariance of $f=f(t)$ implies 
\be{howitvar2}
0=\frac{d f}{dt}|_{t=0}= R \nabla R+ I \nabla I \begin{pmatrix} A \vec v_R-B \vec v_I \cr A \vec v_I + B \vec v_R \end{pmatrix}. 
\ee

\begin{color}{black}
The idea of the proof is to show that if the  matrix $F:=A+iB$ satisfies (\ref{howitvar2}) for all possible vectors $\vec v$, then it must be in the   Lie algebra  spanned by $S_{x,y,z}$, plus the $i \times$ identity. To show this,  we will first compute (\ref{howitvar2}) for   special vectors.
\end{color}

Let us consider the case of vectors for which  $\vec v_I=0$. A direct calculation using the definitions (\ref{X2X3X4}) shows that $I=0$ so that (\ref{howitvar2}) simplifies as 
\be{simp}
\dot f(0)=R \nabla R \begin{pmatrix} A \vec v_R \cr B \vec v_R \end{pmatrix} =0.   
\ee
Using the definitions (\ref{X2X3X4}) we get 
\be{star}
R=X_{3R}^2 - X_{3I}^2 -4 \left(X_{2R}X_{4R}-X_{2I} X_{4I}\right), 
\ee
and an explicit calculation, using the constraint that $\vec v_I=0$,  gives $X_{2I}=X_{3I}=X_{4I}=0$, so that we have 
\be{trentaquattro}
\nabla R(\vec v_R, 0) =2 X_{3R} \nabla X_{3R}  -4 X_{4R} \nabla X_{2R} - 4 X_{2R} \nabla X_{4R}. 
\ee
Now we  specialize further the vector $\vec v$  in (\ref{simp}). 

\begin{enumerate}

\item Set $\hat c_2=c_2=\hat c_1=c_1=0$.  
We have $X_2=X_4=0$ and $X_{3R}=c_{0R} c_{3R}$.    A  direct calculation gives 
$$
\nabla X_{3R}=\nabla \left( c_{0R}c_{3R} - c_{0I} c_{3I}  + \frac{\hat c_{1I} \hat c_{2I}}{3} - \frac{\hat c_{1R} \hat c_{2R}}{3} \right)=
$$
$$
\left(c_{3R},  \, -\frac{\hat c_{2R}}{3},   \,  -\frac{\hat c_{1R}}{3},   \, c_{0R} \, -c_{3I},  \, \frac{\hat c_{2I}}{3},  \, \frac{\hat c_{1I}}{3},   \,  -c_{0I}    \right),
$$
which, using  $\hat c_1=\hat c_2=0$ along with  $c_{0I}=c_{3I}=0$, gives 
\be{ro}
\nabla X_{3R}=(c_{3R}, \,  0, \,   0, \,   c_{0R}, \,   0, \,   0, \,   0, \,   0 ) .
\ee 
Placing this and $X_{3R}=c_{0R} c_{3R}$  in (\ref{simp}), we have 
\be{relat2}
\dot f(0)= c_{0R}c_{3R} \begin{pmatrix} c_{3R} & 0 & 0 & c_{0R} \end{pmatrix} A \begin{pmatrix} c_{0R} \cr 0 \cr 0 \cr c_{3R} \end{pmatrix}=0.  
\ee 
Assume $c_{0R} c_{3R} \ne 0$. Using the fact that $A$ is skew-symmetric, and using 
$c_{0R}^2 \ne c_{3R}^2$, this relation implies  $a_{1,4}=a_{4,1}=0$. 

\item Set $c_0=c_2=0$. 
A direct calculation using (\ref{X2X3X4})  gives $X_3=0$ and  (\ref{star}) gives  $R=4 c_{1R}^3 c_{3R}$. We have  using (\ref{trentaquattro}) and $X_3=0$
\be{weh}
\nabla R=-4 X_{4R} \nabla X_{2R}  -4 X_{2R} \nabla X_{4R}= -4c_{1R} c_{3R} \nabla X_{2R} +4 c_{1R}^2 \nabla X_{4R}. 
\ee
Now use 
\be{nablaX2R}
\nabla X_{2R}= \nabla \left( c_{0R} \frac{\hat c_{2R}}{\sqrt{3}} - c_{0I} \frac{\hat c_{2I}}{\sqrt{3}} -\frac{\hat c_{1R}^2}{3} + \frac{\hat c_{1I}}{3} \right) =
\ee
$$
\left( \frac{\hat c_{2R}}{\sqrt{3}},  \,  -\frac{2}{3} \hat c_{1R},  \,  \frac{c_{0R}}{\sqrt{3}}, \, 0, \,  - \frac{\hat c_{2I}}{\sqrt{3}}, \,  \frac{2 \hat c_{1I}}{3},  \,  - \frac{c_{0I}}{\sqrt{3}},  \,  0  \right) , 
$$
and 
\be{nablaX4R}
\nabla X_{4R}= \nabla \left(\frac{\hat c_{1R}}{\sqrt{3}} c_{3R} - \frac{\hat c_{1I}}{\sqrt{3}} c_{3I} - \frac{\hat c_{2R}^2}{3} + \frac{\hat c_{2I}^2}{3}      \right) =
\ee
$$
\left(  0 , \,  \frac{c_{3R}}{\sqrt{3}} , \,  - \frac{2}{3} \hat c_{2R} , \,  \frac{\hat c_{1R}}{\sqrt{3}} , \,  0 
,\,  -\frac{c_{3I}}{\sqrt{3}} , \,  \frac{2}{3} \hat c_{2I} , \,  - \frac{\hat c_{1I}}{\sqrt{3}}     \right) , 
$$
with $c_0=\hat c_2=0$ and $\hat c_{1I}=c_{3I}=0$, in (\ref{weh}). We get 
$$
\nabla R= \frac{4}{\sqrt{3}} \hat c_{1R}^2\left(0 , \, c_{3R}, \, 0 , \, \frac{\hat c_{1R}}{3}, \, 0 ,\, 0 , \, 0 ,\, 0  \right). 
$$
Placing this in (\ref{simp}) , we get 
\be{fdotz}
\dot f(0)=\frac{16}{9} \hat c_{1R}^5 c_{3R} \begin{pmatrix} 0 & c_{3R} & 0 & \frac{\hat c_{1R}}{3} \end{pmatrix} A 
\begin{pmatrix} 0 \cr \hat c_{1R} \cr 0 \cr c_{3R} \end{pmatrix} =0.  
\ee
This gives  using $a_{1,4}=a_{4,1}=0$, $(c_{3R}^2 - \frac{\hat c_{1R}^2}{3})a_{2,4}=0$ which implies $a_{2,4}=0$, if we choose $\hat c_{1R}$ and $c_{3R}$ different from zero.  
\end{enumerate}

This shows that every matrix $F=A+iB$ in the Lie algebra corresponding to the Lie group which leaves $\tau$ unchanged has to be such that $a_{2,4}=a_{1,4}=a_{4,2}=a_{4,1}=0$. Since we assume that $S_y$ in (\ref{Sy}) also belongs to such Lie algebra and $[S_y, A]$ is real (and skew-symmetric) while $[S_y, B]$ is purely imaginary (and symmetric), imposing this condition on $[S_y, A]$, shows that we must have also $a_{1,3}=a_{3,1}=0$ and  $a_{3,4}=-a_{4,3}=\frac{\sqrt{3}}{2} a_{2,3}$.  Furthermore, imposing the condition that the $(1,3)$ component is zero to $[S_y, A]$, we also get $a_{1,2}=a_{3,4}$. This shows that  the real part of $A+iB$ must be a multiple of $S_y$ in (\ref{Sy}). 

Now consider the restrictions on the $B$ symmetric matrix. Since, with $S_z$ in  (\ref{Sz}), $[S_z, i B]$ is real,  
 it  must be proportional to $S_y$. From this restriction,  it follows that $B$ must be the sum of a diagonal matrix and a matrix proportional to $S_x$ in (\ref{Sx}). Then, considering the Lie bracket $[S_x,i B]$ which must also be proportional to $S_y$ it follows that the diagonal part of $B$ must be a linear combination of the identity and $S_z$. This concludes the proof. 

\end{proof}

\subsection*{Proof of Proposition \ref{onlyone2}}

\begin{proof} We use the notations of the proof of Proposition \ref{onlyone1}. 
Let us first  consider the function $\det(\rho_A)$ in (\ref{detroA}) as 
we act on the vector $\vec v$ as defined in the previous 
proof of Proposition \ref{onlyone1} with $e^{Ft}$. That is, similarly to (\ref{howitvar}), we define a function 
\be{howitvar3}
g(t)=\det(\rho_A)\left( e^{Ft} \vec v \right) =
\left| X_3(e^{Ft} \vec v) \right|^2+ 2\left| X_2(e^{Ft} \vec v) \right|^2+2 \left| X_4(e^{Ft} \vec v) \right|^2, 
\ee 
and we calculate 
$
\frac{d}{dt}|_{t=0} g(t)
$. Notice that this is not set to zero yet. Similarly to what was done in (\ref{howitvar2}), we have 
\be{howitvar4}
\frac{d}{dt}|_{t=0} g(t)= (\nabla |X_3|^2+2 \nabla |X_2|^2 + 2 \nabla |X_4|^2) \begin{pmatrix} A \vec v_R-B \vec v_I \cr A \vec v_I + B \vec v_R \end{pmatrix}. 
\ee
Choose now an initial point so that $\vec v_I=0$ which implies that the imaginary parts of $X_2,$ $X_3$ and $X_4$ are also zero. Therefore (\ref{howitvar4}) gives 
\be{howitvar5}
\frac{d}{dt}|_{t=0} g(t)= (2 X_{3R} \nabla X_{3R} + 4 X_{2R} \nabla X_{2R} + 4 X_{4R} \nabla X_{4R}) \begin{pmatrix} A \vec v_R \cr B \vec v_R \end{pmatrix}. 
\ee 

We now proceed considering subcases as in Proposition \ref{onlyone1},   in fact, the same subcases. 

\begin{enumerate}

\item Set  $\hat c_1=\hat c_2=0$. In this case $X_2=X_4=0$ and $\nabla X_{3R}$ was already computed in (\ref{ro}). Using this expression and the expression of $X_{3R}$,  which in this case is $X_{3R}=c_{0R}c_{3R}$, we obtain 
\be{kappa}
\frac{d}{dt}|_{t=0} g(t)= 2 c_{0R}c_{3R} K, 
\ee 
where (cf. (\ref{relat2}))  
$$
K:= \begin{pmatrix} c_{3R} & 0 & 0 & c_{0R} \end{pmatrix} A \begin{pmatrix} c_{0R} \cr 0 \cr 0 \cr c_{3R} \end{pmatrix}. 
$$
With the given definitions, the expression of the pairwise entanglement (\ref{pairW}) as a function of $t$ is 
$$
\tau_{AB}(t)= 2 \left( g(t) - \sqrt{2 f(t)}\right). 
$$
The condition $\frac{d\tau_{AB}}{dt}|_{t=0}$ gives  
\be{maincond}
\dot g(0)- \frac{ \dot f(0)}{\sqrt{2f(0)}}=0, 
\ee
which using (\ref{kappa}) and (\ref{howitvar2}) and (\ref{relat2}) and the expression of $f(0)=\frac{1}{2} c_{0R}^2 c_{3R}^2$,  gives 
$$
2c_{0R}c_{3R} |c_{0R}c_{3R}| K-  c_{0R}c_{3R} K=0. 
$$
The quantity $c_{0R}c_{3R}$ can be chosen so that this implies  $K=0$ and, as in the proof of Proposition \ref{onlyone1}, this gives $a_{14}=a_{4,1}=0$.

\item Set $c_0=c_2=0$. 

This gives $X_3=0$. We also have $X_{2R}=-c_{1R}^2=-\frac{\hat c_{1R}^2}{3}$, $X_{4R}=c_{1R} c_{3R}=\frac{\hat c_{1R}}{\sqrt{3}} c_{3R}$. The quantities $\nabla X_{2R}$ and $\nabla X_{4R}$ in this case were calculated in (\ref{nablaX2R}) (\ref{nablaX4R}). These formulas give, with $c_0=c_2=0$, 
\be{allfor1}
\nabla X_{2R}=\left( 0, -\frac{2}{3} \hat c_{1R}, 0,   0, 0 , 0 , 0 , 0 \right), 
\ee
\be{allfor2}
\nabla X_{4R}=\left( 0, \frac{c_{3R}}{\sqrt{3}} , 0, \frac{\hat c_{1R}}{\sqrt{3}}, 0 , 0 , 0 , 0 \right). 
\ee
 Using these formulas in (\ref{howitvar5}), we obtain 
 \be{ddt}
 \frac{d}{dt}|_{t=0} g(t)= \frac{4}{3} \hat c_{1R} \left( 0 , c_{3R}^2+\frac{2}{3} \hat c_{1R}^2, 0, \hat c_{1R} c_{3R}, 0, 0, 0, 0 \right)\begin{pmatrix} A \vec v_R \cr B \vec v_{R}  \end{pmatrix}. 
 \ee
 In this case, $X_3-4X_2X_4= 4c_{1R}^3 c_{3R}=\frac{4}{3} \frac{\hat c_{1R}^3}{\sqrt{3}} c_{3R}$, and therefore $f(0)=\frac{1}{2} |X_3-4X_2 X_4|^2=\frac{8}{27} \hat c_{1R}^6 c_{3R}^2$.  
 Replacing this in (\ref{maincond}),  together with the expression of $\dot f(0)$ calculated in (\ref{fdotz}), we get 
 \be{REP}
 \frac{4}{3} \hat c_{1R} \left(0, c_{3R}^2+ \frac{2}{3} \hat c_{1R}^2, 0, \hat c_{1R}c_{3R} \right) A \begin{pmatrix} 0 \cr \hat c_{1R} \cr 0 \cr c_{3R} \end{pmatrix} 
- \frac{4}{3} \sqrt{3} \frac{\hat c_{1R}^5 c_{3R}}{|c_{1R}^3 c_{3R}|} (0, c_{3R}, 0,\frac{\hat c_{1R}}{3}) A  \begin{pmatrix} 0 \cr \hat c_{1R} \cr 0 \cr c_{3R} \end{pmatrix} =0. 
 \ee
Now choose $c_{3R}=\hat c_{1R}$.  Using this in (\ref{REP}), we obtain,  after simplifications, 
$$
\begin{pmatrix} 0 &  5-3 \sqrt{3} & 0 & 3 -\sqrt{3} \end{pmatrix} A \begin{pmatrix} 0 \cr 1 \cr 0 \cr 1 \end{pmatrix}=0 
$$ 
This, using $a_{1,4}=a_{4,1}=0$,  implies $a_{2,4}=a_{4,2}=0$. 
\end{enumerate}

The rest of the proof proceeds as  the proof of Proposition \ref{onlyone1}.

\end{proof}

\end{document}